\documentclass{article}
\usepackage[textsize=footnotesize]{todonotes}
\usepackage{amsmath}
\usepackage{amsthm}
\usepackage{amsfonts}
\usepackage{algorithm}
\usepackage{algorithmic}
\usepackage{esint}

\usepackage{array}
\usepackage[affil-it]{authblk}
\usepackage{bm}
\usepackage{bbm}
\usepackage{booktabs}
\usepackage{color}
\usepackage[OT1]{fontenc}
\usepackage{graphicx}
\usepackage{float}
\usepackage{bigints}

\usepackage{multirow}
\usepackage{subcaption}

\usepackage{mathrsfs}
\usepackage[bbgreekl]{mathbbol}
\usepackage{mathabx}
\usepackage{mathtools}
\usepackage{natbib}
\usepackage[section]{placeins}
\usepackage{subcaption}
\usepackage{verbatim}
\usepackage{enumitem}
\usepackage{algorithm}
\usepackage{algorithmic}

\usepackage{xr}
\usepackage{xcolor}

\usepackage{optidef}

\RequirePackage[colorlinks,citecolor=blue,urlcolor=blue]{hyperref}

\DeclareSymbolFontAlphabet{\mathbbm}{bbold}
\DeclareSymbolFontAlphabet{\mathbb}{AMSb}%

\usepackage{setspace}
\newtheorem{theorem}{Theorem}[section]
\newtheorem{lemma}[theorem]{Lemma}

\newtheorem{corollary}[theorem]{Corollary}
\newtheorem{definition}[theorem]{Definition}
\newtheorem{remark}[theorem]{Remark}

\numberwithin{table}{section}
\numberwithin{figure}{section}

\newcommand{\mE}{\mathbb{E}}

\newcommand{\frakF}{\mathfrak{F}}

\newcommand{\calG}{\mathcal{G}}

\newcommand{\calC}{\mathcal{C}}

\newcommand{\calK}{\mathcal{K}}
\newcommand{\calN}{\mathcal{N}}

\newcommand{\mP}{\mathbb{P}}

\newcommand{\vw}{\breve{w}}
\newcommand{\vW}{\breve{W}}
\newcommand{\vZ}{\breve{Z}}
\newcommand{\vS}{\breve{S}}

\newcommand{\tr}{\operatorname{tr}}

\begin{document}

\title{Adaptable Fingerprinting with Nonlinear Shrinkage for Climate Change Detection and Attribution under Variance Heterogeneity}

\author{
{Haoran Li}\thanks{Email: hzl0152@auburn.edu}  { and Yan Li}\thanks{Email: yzl0317@auburn.edu}\\
Department of Mathematics and Statistics, Auburn University, Auburn, AL, 36849
}
\maketitle

\begin{abstract}
Detection and attribution (D\&A) of climate change relies on \emph{fingerprinting}—a linear errors-in-variables regression framework in which both predictors and responses exhibit internal variability governed by a proportional covariance structure, subject to a variability inflation factor. Accurate estimation of the scaling factors (regression coefficients) depends on inferring the precision matrix of the regression errors from limited climate model control runs. In high-dimensional settings, existing approaches often overlook the variance inflation of the predictors and suffer from imprecise precision matrix estimates, yielding biased estimators, underestimated uncertainties, and confidence intervals with poor coverage. We propose a nonlinear, rotation-invariant shrinkage framework for estimating the precision matrix that restores the asymptotic optimality of the \emph{total least squares} estimator in high-dimensional regimes. Our procedure jointly estimates the scaling factors and the variability inflation factor, thereby correcting estimation bias, and incorporates consistent variance estimators to enable valid uncertainty quantification. We also develop a residual consistency test to assess model adequacy. Numerical studies demonstrate precise estimation, improved confidence interval coverage, and higher efficiency. Applied to annual mean near-surface air temperature data from 1951–2020, our method produces narrower and more reliable confidence intervals, yielding refined attribution results.
\end{abstract}

\noindent%
\vfill

\allowdisplaybreaks

\section{Introduction}
\label{sec:intro}
Detection and attribution (D\&A) analyses of climate variables form the statistical foundation for key conclusions of the Intergovernmental Panel on Climate Change (IPCC), including the finding that more than half of the observed warming in recent decades is attributable to anthropogenic forcings \citep{hegerl2007understanding, Bind:etal:dete:2013, eyring2021human}. In this framework, \emph{detection} refers to determining whether observed climate changes are statistically significant, whereas \emph{attribution} quantifies the contributions of specific external forcings with associated confidence \citep{Hege:Zwie:use:2011}. By comparing climate model simulations with observed responses, D\&A evaluates whether the observed changes are consistent with model-simulated responses, often referred to as \emph{fingerprints} or \emph{signals}, to external forcings such as greenhouse gases, aerosols, solar variability, and volcanic activity.

Optimal fingerprinting (OF), the standard statistical method for D\&A, formulates a multivariate regression of observed climate responses onto signals of external forcings \citep{hegerl1996detecting, Alle:Tett:chec:1999, Alle:Stot:esti:2003}. The primary inferential targets are the regression coefficients, also known as \emph{scaling factors}, which quantify the amplitudes of the signals. Statistically significant scaling factors indicate detection, and confidence intervals that include one support attribution. Under the idealized assumption that the precision matrix $\Sigma^{-1}$ is known, OF is optimal in the sense that the generalized least squares (GLS) estimator, after prewhitening by $\Sigma^{-1}$, yields the best linear unbiased estimator (BLUE). Subsequent work recognized that the forcing signals are not directly observed but estimated from finite model ensembles, introducing measurement error in both predictors and responses. This leads naturally to a linear errors-in-variables (EIV) formulation. When the covariance structures of model simulations and observations are assumed to be proportional, with the simulation variance inflated by a constant factor, a weighted total least squares (TLS) estimation becomes justified, as both predictors and responses can be prewhitened using $\Sigma^{-1}$ \citep{Alle:Stot:esti:2003}.

In practice, $\Sigma$ is unknown and must be estimated from a limited number of control simulations, often in high-dimensional settings where the number of spatial and temporal dimensions exceeds the number of available control runs. In such cases, the sample covariance matrix is unstable or singular, rendering it unsuitable for prewhitening. Early approaches addressed this by projecting the data onto a truncated basis of empirical orthogonal functions (EOFs) \citep{hegerl1996detecting, Alle:Tett:chec:1999}, although this can suppress important spatial structure. More recent approaches incorporate linear shrinkage estimators, most notably the Ledoit--Wolf estimator \citep{ledoit2004well}, into the regularized optimal fingerprinting (ROF) framework \citep{ribes2009adaptation,li2025regularized}, improving numerical stability and finite-sample performance. Alternative approaches, including Bayesian model averaging \citep{katzfuss2017bayesian} and integrated likelihood methods \citep{hannart2016integrated}, aim to seek greater robustness to covariance misspecification but remain computationally intensive and sensitive to prior choices. 
Another line of work uses estimating equations with structural constraints, such as block Toeplitz forms \citep{ma2023optimal}, to improve coverage, but these methods may lose efficiency due to suboptimal weighting and the lack of direct temporal modeling.

Uncertainty in estimating $\Sigma$ has important implications for regularized optimal fingerprinting, yet the theoretical properties of the resulting estimators remain incompletely understood. For instance, confidence intervals based on asymptotic normal approximations \citep{fuller1980properties, gleser1981estimation} or bootstrap methods \citep{pesta2012total} tend to be overly narrow when $\Sigma$ is known only up to a scale \citep{delsole2019} or entirely unknown \citep{li2021confidence, li2023regularized}. This raises a fundamental question: \emph{when the covariance matrix must be estimated, do inference procedures developed under the assumption of a known covariance remain valid and optimal, or do they suffer from systematic inefficiency--and if so, can optimality be recovered by explicitly accounting for covariance estimation uncertainty?}

Beyond the challenge of estimating $\Sigma$, an additional complication arises from \emph{model--observation structural differences} in internal variability. In practice, climate models often simulate internal variability that differs in magnitude or structure from observations, reflecting limitations in representing key physical processes and feedbacks \citep[see, e.g.,][]{imbers2014sensitivity, jain2023importance, ma2023optimal}. The IPCC Assessment Reports also note that such discrepancies can affect the reliability of detection and attribution results \citep{eyring2021human}. When the model variability is misspecified, the prewhitening step can distort the scaling of externally forced signals, leading to misestimation of the scaling factors. However, the existing ROF framework generally overlooks this form of structural difference.

To address these challenges, we propose a new ROF framework featuring two main innovations. First, we introduce a nonlinear, rotation-invariant covariance estimator that regularizes the eigenvalue spectrum of the sample covariance matrix through flexible shrinkage functions, with the optimal shrinkage selected in a data-driven manner by minimizing the asymptotic mean squared error. This design yields consistent variance estimators and valid confidence intervals, providing theoretically grounded and computationally efficient inference in high-dimensional EIV settings. Second, we extend the classic fingerprinting model to incorporate heterogeneous inflation structures that account for potential model–observation structural differences in internal variability, relaxing the conventional assumption of a shared covariance matrix and enhancing robustness to model misspecification. In addition, we introduce a residual consistency test for assessing model adequacy in high-dimensional regimes, a critical point that has not received sufficient attention in the existing literature. We establish theoretical properties of the proposed method and demonstrate its empirical advantages via asymptotic analysis and simulation studies. As shown in our numerical experiments, the proposed point estimators are unbiased with substantially reduced root mean squared error, and the associated confidence intervals achieve coverage rates close to the nominal level.

The remainder of the paper is organized as follows. Section~\ref{sec:review_OF} introduces the fingerprinting framework and the technical assumptions for deriving theoretical results. Section~\ref{sec:methodology} introduces our proposed method for estimation. Section~\ref{sec:numerical} reports simulation results. Section~\ref{sec:application} applies the method to detection and attribution of annual mean temperature over 1951--2020 at continental and subcontinental scales. Section~\ref{sec:disc} concludes with a discussion of implications and potential extensions.

\section{Problem Setup}
\label{sec:review_OF}
The prevailing fingerprinting framework is formulated as a linear error-in-variables (EIV) regression model, with climate preindustrial control simulation runs facilitating covariance estimation. Specifically, the observed data consist of three independent components modeled as
\begin{align}
&Y = \sum_{i=1}^p X_i \beta_i + \epsilon, \quad \mbox{where} \quad \epsilon \sim \mathcal{N}(0,\Sigma); \label{eq:fingerprint1} \\
&\tilde{X}_{ik} = X_i + \eta_{ik}, \quad \mbox{where} \quad \eta_{ik} \stackrel{iid}{\sim} \mathcal{N}(0, \sigma^2_{X}\Sigma),   \quad k = 1, \dots, n_i, \quad i = 1, \dots, p; \label{eq:fingerprint2}\\
&Z_j \stackrel{iid}{\sim} \mathcal{N}(0, \sigma^2_{Z}\Sigma), \quad j=1,2,\dots, m.\label{eq:fingerprint3}
\end{align}

Firstly, the observed climate response \(Y \in \mathbb{R}^{N}\) is modeled as a linear combination of unobserved fingerprints \(X_i \in \mathbb{R}^N\), \(i = 1, 2, \dots, p\), corresponding to \(p\) external forcings. The coefficients \(\beta = (\beta_1, \beta_2, \dots, \beta_p)^\top\in \mathbb{R}^p \) represent the unknown scaling factors. The error term \(\epsilon\) accounts for internal climate variability and is assumed to follow a multivariate normal distribution with mean zero and positive definite covariance matrix \(\Sigma\). The primary objective is to estimate the scaling factors \(\beta\) and to provide associated uncertainty quantification.

Secondly, for each forcing \(i\), an ensemble of noisy realizations \(\tilde{X}_{ik}\in\mathbb{R}^N,\) $k=1,\dots, n_i$, of the true fingerprint \(X_i\) is observed, with ensemble size $n_i\geq 2$. 
Each realization \(\tilde{X}_{ik}\) is contaminated by measurement error \(\eta_{ik}\), whose variability is assumed to be proportional to $\Sigma$ through a scale inflation factor $\sigma_X^2 > 0$. Classical D\&A analyses often set $\sigma_X^2 = 1$, implicitly assuming that model-simulated variability matches that of observations. In practice, however, ensemble simulations frequently exhibit inflated variability, and misspecification of $\sigma_X^2$ can bias estimation of $\beta$. In this paper, we therefore allow $\sigma_X^2$ to be estimated in a data-driven manner rather than fixed \emph{a priori}.

Thirdly, a set of \(m\) preindustrial control runs \(Z_j \in \mathbb{R}^N\), \(j = 1, \dots, m\), is simulated from climate models in the absence of external forcing. It is standard to assume that the mean of the runs is $0$ and their variability matches that of the observations up to a scale inflation factor \(\sigma^2_{Z} > 0\). In many applications, the number of available runs $m$ is often comparable to the dimension $N$, a regime in which high-dimensional covariance estimation plays a critical role. Similar to the treatment of $\sigma_X^2$, classical D\&A analysis often assume that $\sigma^2_Z$ is known. While prior knowledge of $\sigma_Z^2$ can facilitate the estimation of $\sigma_X^2$, in this paper, we consider both scenarios where $\sigma_Z^2$ is specified as \emph{a priori} and where it must be estimated from the data.

Naturally, the ensemble mean 
\[\tilde{X}_i = \frac{1}{n_i} \sum_{k=1}^{n_i} \tilde{X}_{ik}\]
serves as an estimator for the true fingerprint \(X_i\). In practice, however, the ensemble size $n_i$ is typically much smaller than $N$, rendering the estimator noisy. For example, \cite{Alle:Tett:chec:1999} and \cite{hegerl2007understanding} report that in typical applications, $n_i$'s are about $50$, while $N$ is on the order of $1,000$. Technically, when $n_i = o(N)$, the expected total estimation error remains non-negligible and satisfies $\mE\|\tilde{X}_i - X_i\|^2_2 = O(N/n_i)$, under standard assumptions on $\Sigma$ (see \ref{enum:boundedness} for details).

The sample covariance matrix computed from the control runs 
\[S = \frac{1}{m} \sum_{j=1}^m Z_j Z_j^T \] 
serves as a rudimentary estimator of $\Sigma$ up to the inflation factor $\sigma_Z^2$. 

For convenience, we introduce the following notation
\[
\tilde{X} = (\tilde{X}_1, \tilde{X}_2, \dots, \tilde{X}_p) \in \mathbb{R}^{N\times p}, \quad X = (X_1, \dots, X_p) \in \mathbb{R}^{N\times p}, \quad D = \operatorname{Diag}(1/n_1, \dots, 1/n_p).
\]
The notation $\|\cdot\|_2$ denotes the spectral norm (operator $\ell_2$-norm) of a matrix.
For a symmetric $N\times N$ matrix $A$, let $\lambda_i(A)$ denote its $i$th largest eigenvalue, with the conventions $\lambda_{\min}(A) = \lambda_{N}(A)$ and $\lambda_{\max}(A) = \lambda_1(A)$. Define the \emph{empirical spectral distribution} (ESD) $F^A$ of $A$ by
\[ F^A(\tau) = \frac{1}{N} \sum_{i=1}^N \mathbbm{1}(\lambda_i(A)\leq \tau).\]

Throughout the paper, the following technical assumptions are imposed.

\begin{enumerate}[label ={\bf C\arabic*}]\itemsep0.3em
\item \label{enum:HD_regime} (High-dimensional settings) While the number of predictors $p$ is fixed, we assume $N, m \to \infty$ simultaneously such that $\sqrt{N} | N/m - \gamma|  \to 0 $ for some  $\gamma \in (0,\infty)$. We further assume that there exist constants $ l_1>0$ and $0< l_2 < 1/2$ such that $(\log N)^{l_1}\leq n_i \leq N^{l_2}$, for all $i =1,2,\dots, p$.
\item \label{enum:boundedness} (Boundedness of the spectrum of $\Sigma$) $0< \liminf_{N\to\infty} \lambda_{\min}(\Sigma) \leq  \limsup_{N\to \infty} \lambda_{\max}(\Sigma)<\infty$.
\item \label{enum:boundedness2} (Boundedness of the spectrum of $X$) $0<\liminf_{N\to\infty} \lambda_{\min}(X^\top X/N) \leq \limsup_{N\to\infty} \lambda_{\max}(X^\top X/N)<\infty$.
\item \label{enum:stability_ESD} (Asymptotic stability of the spectrum of $\Sigma$) There exists a distribution $L^{\Sigma}$ with compact support in $(0, \infty)$ such that $\sqrt{N} \|F^{\Sigma}(x) - L^{\Sigma} (x)\|_\infty \to 0$. Here, $\|\cdot \|_\infty$ stands for the uniform norm of a function. 
\item \label{enum:converge_X_Sigma_1} (Asymptotic stability of $X$) We assume the projection of $X$ onto the eigenspace of $\Sigma$ stabilizes as $N\to\infty$ in the following sense: for any continuous function $g:\, \mathbb{R}\to \mathbb{R}$, there exists a $p\times p $ matrix $M(g)$, such that 
\[ \frac{1}{N} X^\top g(\Sigma)X \longrightarrow M(g), \quad \mbox{as }N\to\infty.\]
Here, $g(\Sigma) = Q \operatorname{Diag}(g(\lambda_1),\cdots, g(\lambda_N))Q^\top$, where $Q$ is the eigenvector matrix of $\Sigma$  associated with the ordered eigenvalues $\lambda_1\geq \cdots\geq \lambda_N$. 
\end{enumerate}

The $o(N^{-1/2})$ convergence rates of $N/m$ and $F^{\Sigma}$ assumed in \ref{enum:HD_regime} and \ref{enum:stability_ESD} are imposed solely for technical convenience. In principle, this condition can be removed by appropriately modifying the arguments in the proofs; however, doing so would substantially complicate the exposition and obscure the main ideas. For clarity and conciseness, we therefore maintain this assumption throughout the paper. Notably, the convergence rate assumption is practically not overly restrictive, as it imposes little constraint on the observations.

The following lemma, a standard result in random matrix theory \citep[e.g.,][]{yin1988limit}, bounds the extreme eigenvalues of \(S\).
\begin{lemma}\label{lemma:extreme_eigenvalue_bound}
Suppose that \ref{enum:HD_regime}--\ref{enum:stability_ESD} hold.  For any constant $c>0$, as $N,m \to\infty$, 
\[ \mP\left( \lambda_1(S) \leq  c+ \limsup_{N\to\infty}\sigma^2_{Z}\lambda_{\max}(\Sigma)  (1+\sqrt{\gamma})^2 \right) \to 1.  \]
\end{lemma}
\noindent This result implies that, with high probability, all eigenvalues of $S$ are contained within a compact interval.

\section{Methodology}\label{sec:methodology}
Under the framework of EIV, the method of \emph{total least squares} (TLS) is one of the broadly used estimation techniques that accounts for errors in both predictors and the response. Under the model specification \eqref{eq:fingerprint1}-\eqref{eq:fingerprint2}, when the inflation factor $\sigma^2_{X}$ is known, the TLS estimator of $\beta$ with a weight matrix $W$ is obtained as
\begin{equation}\label{eq:def_TLS}
\hat{\beta} = \hat{\beta}(\sigma^2_{X}, W) = \arg \min_{\beta\in \mathbb{R}^p } \frac{\|W^{-1/2} (Y - \tilde{X} \beta)\|_2^2}{1 + \sigma^2_{X} \beta^\top D \beta}.
\end{equation}
The formula is adapted from standard results in the errors-in-variables literature \citep[e.g.,][]{fuller2009measurement,carroll2006measurement}, and the derivation is omitted for brevity.
 The matrix $W$ is expected to be an estimator of the covariance matrix $\Sigma$ up to a inflation factor, capturing the covariance structure of the internal variability in both $Y$ and $\tilde{X}$. In particular, when $W = c\Sigma$ for some $c>0$, the residuals $Y - \tilde{X}\beta$ are effectively whitened. Moreover, if $\sigma^2_{X} =0$, meaning the predictors are measured without error, the TLS estimator reduces to the GLS estimator with the weight matrix $W$. Indeed, $\hat{\beta}$ admits a closed-form expression, detailed in Lemma~\ref{lemma:explicit_formula_of_beta} in the Appendix.

Recall that the sample covariance matrix $S$ serves as a rudimentary estimator of \(\sigma_Z^2\Sigma\). Under the classical asymptotic regime, where the dimension \(N\) is fixed and \(m \to \infty\), \(S\) consistently estimates \(\sigma^2_{Z}\Sigma\). However, when \(N\) is comparable to $m$ as in \ref{enum:HD_regime}, the performance of \(S\) deteriorates significantly; in particular, when $N>m$, $S$ is singular and cannot be used directly as the weight matrix. In the literature, the asymptotic behavior of \(S\) under \ref{enum:HD_regime} has been extensively studied. Importantly, in this regime, although inconsistent, \(S\) still retains substantial information about \(\Sigma\) and is therefore valuable for whitening the covariance in both \(Y\) and \(\tilde{X}\) (see, e.g., \citet{bai2010spectral}). A detailed discussion of these results is beyond the scope of this work.

To address the problem of covariance estimation, \citet{ribes2009adaptation} introduced a linear shrinkage framework in which $W$ is taken as a linear function of $S$, namely $wS+ I_N$ with a regularization parameter $w>0$. This idea originates from the seminal work of \citet{ledoit2004well}, which introduced linear shrinkage for improved estimation of covariance matrices. More recently, \citet{li2025regularized} refined this approach in the context of TLS estimation and showed that the choice of regularization parameter substantially influences estimation accuracy. Meanwhile, it is well recognized that the optimal shrinkage function $f$, for minimizing the Frobenius norm between $f(S)$ and $\Sigma$, is inherently nonlinear \citep{ledoit2017nonlinear}. 

Motivated by this, we consider a more flexible family of polynomial shrinkage functions of the form
\begin{equation} \label{eq:poly}
f(S) = \ell_k S^k + \ell_{k-1} S^{k-1} + \cdots + \ell_1 S + I_N,
\end{equation}
where the degree $k>0$ is user-specified and fixed, and the coefficients $(\ell_k, \ell_{k-1},\dots,\ell_1)$ serve as regularization parameters, with $\ell_k\neq 0$. Here, the coefficient of the intercept term is fixed as $1$ without loss of generality, since the TLS estimator is invariant if $W = f(S)$ is scaled by any nonzero constant. The linear shrinkage framework is clearly a special case when $k=1$. To ensure that $f(S)$ is well-conditioned and suitable for subsequent estimation, we impose structural constraints on the coefficients. These constraints and their motivations are detailed in Section~\ref{subsec:selection_shrinkage}.  Since any smooth function can be well approximated by high-order polynomials, this family of shrinkage
functions is both flexible and practically advantageous due to its simple functional form. While the proposed methods are applicable to any $k>0$, based on extensive simulation studies, we recommend setting \(k = 3\), which offers a favorable balance between stability, flexibility, and computational efficiency.

While polynomial shrinkage is suitable for general purposes, in specific scenarios where additional structural information about $\Sigma$ is available, specialized shrinkage forms may offer improved performance. The proposed framework can be extended to accommodate a broader class of analytic shrinkage functions of $S$. Details of this extension are provided in Section~\ref{sec:extension_general_shrinkage} of the Appendix.

The remainder of this section is organized as follows.
Section~\ref{subsec:selection_shrinkage} considers the setting in which the inflation factors $\sigma_X^2$ and $\sigma_Z^2$ are treated as known. Under this setting, we establish the asymptotic normality of the TLS estimator for any fixed admissible shrinkage function $f$, and develop a data-driven procedure for selecting $f$ by minimizing the asymptotic total variance of the TLS estimator. Section~\ref{subsec:estimation_a} addresses the estimation of $\sigma_X^2$ and $\sigma_Z^2$ and establishes the asymptotic normality of the TLS estimator when these factors are estimated from data. Section~\ref{subsec:CI_Diagnosis} presents the construction of confidence regions for the true scaling factor $\beta$ and introduces a residual diagnostic procedure for evaluating model adequacy.

\subsection{Asymptotic Normality and Shrinkage Selection}\label{subsec:selection_shrinkage}

Consider any constant $\psi_1$ such that 
\[\psi_1 > \limsup_{N\to\infty}\sigma^2_{Z}\lambda_1(\Sigma)  (1+\sqrt{\gamma})^2.\] 
As established in Lemma \ref{lemma:extreme_eigenvalue_bound}, under our assumptions, all eigenvalues of $S$ lie within the interval $[0, \psi_1]$ with high probability when $N$ is sufficiently large. 
In the subsequent analysis, $\psi_1$ is treated as a fixed constant. In practice, $\psi_1$ can be set as $\xi\lambda_1(S)$, where $\xi>1$ is a user-specified hyperparameter. Our main results are not sensitive to the choice of $\xi$ provided it lies within a reasonable range; the recommended setting of $\xi$ is provided in Section \ref{appendix:sec:implement_details} of the Appendix. 

\begin{definition}\label{def:admissible}
We say a polynomial function $f = \ell_k x^k + \cdots + \ell_1 x +1$ is \emph{admissible} if it satisfies:
\begin{itemize}
    \item[(i)] (Monotonicity and boundedness) $f(x)$ is non-decreasing on $[0, \psi_1]$ and $f(\psi_1) \leq \psi_2$;
    \item[(ii)] (Root separation) $f(x)$ has $k$ distinct roots $r_1,\dots, r_k \in\mathbb{C}$ whose pairwise separation is bounded below by a constant $\psi_3>0$, i.e., $\min_{i\neq j} |r_i - r_j|  \geq \psi_3$.
\end{itemize}    
\end{definition}
\noindent 
Here, $\psi_1$, $\psi_2$, and $\psi_3$ are fixed positive constants controlling, respectively, the upper bound of the eigenvalue domain on which $f$ is defined, the overall magnitude of the shrinkage function, and the minimal spacing among its polynomial roots. Practical settings for $\psi_2$ and $\psi_3$ are provided in Section~\ref{appendix:sec:implement_details} of the Appendix. Notice that the eigenvalues of $f(S)$ are $f(\lambda_j) = f(\lambda_j(S))$, $j=1,\dots, N$. The monotonicity condition ensures that the regularized eigenvalues $f(\lambda_j)$ preserve the original ordering, i.e., $f(\lambda_1) \geq \cdots \geq f(\lambda_N)$. The root separation condition facilitates the estimation of the asymptotic variance of the TLS estimator. This requirement is not overly restrictive since $\psi_3$ can be arbitrarily small, and any polynomial with multiple roots can be well approximated by polynomials with simple roots whose locations converge. 

Basic calculus shows that the reciprocal of any polynomial with distinct roots admits a \emph{partial fraction decomposition}. In particular, suppose $f(x)$ has roots $r_j$, $j=1,2,\dots, k$ (possibly complex-valued). Then, 
\begin{equation}\label{eq:partial_fraction1}
\frac{1}{f(x)} = \frac{1}{\ell_k x^k + \cdots +\ell_1 x + 1} = \sum_{j=1}^k \frac{b_j}{x - r_j},    
\end{equation} 
where the coefficients $b_j$ are given by 
\begin{equation}\label{eq:partial_fraction2}
b_j = \frac{1/\ell_k}{ \prod_{i\neq j} (r_j-r_i)}, \quad j =1,2,\dots, k.
\end{equation}

For any $r, r_1, r_2\in \mathbb{C}\setminus[0, \psi_1]$, define the following objects
\begin{align}
&\mu(r) = \frac{1}{N} \tr [(S - r I_N)^{-1}], \label{eq:def_m_z}\\
&\theta(r) = 1/\{1- N/m - (N/m)r\mu(r)\},\label{eq:def_theta_z} \\ 
&\pi(r) = (m/N) (\theta(r) - 1)\label{eq:def_pi_z}\\
&\delta(r) =  \frac{1}{N} \tilde{X}^T (S - r I_N)^{-1} \tilde{X} - (\sigma^2_X/\sigma^2_Z)\pi(r) D.\label{eq:def_delta_z}\\
&\kappa(r_1, r_2) =\begin{cases}   
\theta(r_1) \theta(r_2) \cfrac{r_1\pi(r_1) - r_2 \pi(r_2)}{r_1 - r_2}, & \mbox{ if }r_1\neq r_2; \\[5pt] 
\theta^2(r_1) \Big[ r_1 \pi'(r_1) + \pi(r_1)\Big], & \mbox{ if }r_1 = r_2;\end{cases} \label{eq:def_kappa_z1z2}\\
&\omega(r_1, r_2) = \begin{cases}
    \theta(r_1)\theta(r_2) \cfrac{r_1\delta(r_1) - r_2 \delta(r_2)}{r_1 - r_2}, & \mbox{ if } r_1 \neq r_2;\\
    \theta^2(r_1)[ \delta(r_1) + r_1\delta'(r_1)], & \mbox{ if }r_1 = r_2.
\end{cases}\label{eq:def_omega_z1z2}
\end{align}
Here, $\pi'(r)$, $\theta'(r)$ and $\delta'(r)$ are the derivatives of $\pi(r)$, $\theta(r)$ and $\delta(r)$, respectively. 

For any admissible polynomial $f$, define 
\begin{align}
&\Delta(f)  = \sum_{j=1}^k b_j \delta(r_j), \label{eq:def_Delta} \\
&\Pi(f) = \frac{1}{\sigma^2_Z} \sum_{j=1}^k b_j \pi(r_j), \label{eq:def_Pi}\\
&\Omega(f) =\frac{1}{\sigma^2_Z} \sum_{j=1}^k \sum_{i=1}^k b_ib_j \omega(r_i, r_j),\label{eq:def_Omega} \\
&K(f)  = \frac{1}{\sigma^4_Z}\sum_{j=1}^k \sum_{i=1}^k b_i b_j \kappa(r_i,r_j).\label{eq:def_K}
\end{align}
where $r_j$'s are the roots of $f$ and $b_j$'s are the associated coefficients given in \eqref{eq:partial_fraction2}. It is worth noting that \(\Delta(f), \Omega(f) \in \mathbb{R}^{p \times p}\), while \(\Pi(f), K(f) \in \mathbb{R}\). The functions depend on the inflation factors $\sigma_X^2$ and $\sigma_Z^2$, but this dependence is suppressed in the notation for simplicity. 

These objects consistently estimate population-level quantities that involve the unknown true fingerprint matrix $X$ and the covariance matrix $\Sigma$, as shown in the following lemma. 
\begin{lemma}
    \label{lemma:convergence_Delta_Pi_Omega_K}
    Suppose that \ref{enum:HD_regime}-\ref{enum:converge_X_Sigma_1} hold. For any admissible polynomial $f$, we have 
    \begin{align*}
    & \Pi(f) - \frac{1}{N}\operatorname{trace}\Big([f(S)]^{-1}\Sigma\Big) = o_p(N^{-1/2}), \\
    & K(f) - \frac{1}{N}\operatorname{trace}\Big([f(S)]^{-1}\Sigma [f(S)]^{-1}\Sigma\Big) = o_p(1),\\
    &\Big\|\Delta(f) - \frac{1}{N} X^T [f(S)]^{-1}X\Big\|_2 = o_p(1),\\
    &\Big\|\Omega(f) - \frac{1}{N} X^T [f(S)]^{-1}\Sigma [f(S)]^{-1} X \Big\|_2  = o_p(1). 
    \end{align*}
\end{lemma}

We now establish the asymptotic normality of the TLS estimator with the weight matrix $f(S)$. 
\begin{theorem}
    \label{thm:normality}
    Suppose \ref{enum:HD_regime}--\ref{enum:converge_X_Sigma_1} hold. Let $f$ be any admissible polynomial, and consider the TLS estimator $\hat{\beta}(f) =  \hat{\beta}(\sigma^2_X, f(S))$ constructed using the true factor $\sigma_X^2$ and the weight matrix $f(S)$. Define
    \begin{equation*}
    \Xi(f) = \Xi(f,\sigma_X^2, \sigma_Z^2) = (1+ \sigma^2_X\beta^T D\beta) \Delta^{-1}(f) \Big\{\Omega(f) + \sigma^2_X K(f) (D^{-1} + \sigma^2_X\beta \beta^T)^{-1} \Big\}\Delta^{-1}(f).
    \end{equation*}
    Then, as $N\to\infty$,
    \[ N^{1/2}  [\Xi(f)]^{-1/2}\Big(\hat{\beta} (f) - \beta\Big)\, \stackrel{D}{\longrightarrow}\, \calN\Big(0, I_p\Big).\]  
\end{theorem}
\begin{remark}
    \label{remark:normality_when_ensemble_size_bounded}
    Theorem~\ref{thm:normality} remains valid even when the ensemble sizes \(n_j\) are bounded (\(n_j < \infty\) for \(j = 1, \dots, p\)) as \(N \to \infty\). 
    This corresponds to the special case \(l_1 = 0\) in \ref{enum:HD_regime}. The regime has been considered in \citet{li2023regularized} and \citet{li2025regularized}.
\end{remark}
\begin{remark}
    \label{remark:reduce_Xi}
    Under \ref{enum:HD_regime},  as $\min_{1\leq i\leq p} n_i \to \infty$, we have $\Xi(f)$ reduces to $\Xi_0(f) = \Delta^{-1}(f) \Omega(f) \Delta^{-1}(f)$ in the sense that $\|\Xi(f) -\Xi_0(f) \|_2 \stackrel{P}{\longrightarrow} 0$. By Slutsky's theorem, the conclusion of Theorem \ref{thm:normality} remains valid if $\Xi(f)$ is replaced by $\Xi_0(f)$. However, in finite samples, $\Xi(f)$ provides a more accurate approximation of the variance of $\hat{\beta}(f)$, particularly when the ensemble sizes $n_j$'s are relatively small. 
 \end{remark}

We now address the problem of selecting shrinkage functions. Let $\frakF$ denote the collection of all admissible polynomials of fixed degree $k$ and the hyperparameters $\psi_1,\psi_2, \psi_3$. 
For any $f\in \frakF$, define 
\begin{equation}\label{eq:def_hat_xi}
\hat{\Xi}(f) =\hat{\Xi}(f, \sigma_X^2, \sigma_Z^2) =  (1+ \sigma^2_X\hat{\beta}^T(f) D\hat{\beta}(f)) \Delta^{-1}(f) \Big\{\Omega(f) + \sigma^2_X K(f) (D^{-1} + \sigma^2_X\hat{\beta}(f) \hat{\beta}^T(f) )^{-1} \Big\}\Delta^{-1}(f).
\end{equation}
The asymptotic total variance of $\hat{\beta}(f)$ is then estimated by $\operatorname{trace}\hat{\Xi}(f)$. Accordingly, we propose the following data-driven selection of the shrinkage function:
\begin{equation}\label{eq:def_hat_f}
\hat{f} = \arg\min_{f\in\frakF} \operatorname{trace} [\hat{\Xi}(f)].
\end{equation}
Note that, by Lemma \ref{lemma:convergence_Delta_Pi_Omega_K}, the objective function is well-conditioned in the sense that there exist constants $0< c_1 < c_2$ depending only on $k, \psi_1, \psi_2, \psi_3$ such that  as $N\to\infty$,
\[\mP\Big( c_1 < \liminf_{f\in \frakF}  \operatorname{trace}[\hat{\Xi}(f)] \leq \limsup_{f\in\frakF} \operatorname{trace}[\hat{\Xi}(f)] <c_2\Big) \longrightarrow 1.\]

Since $\frakF$ is a closed set, the minimizer $\hat{f}$ exists. However, uniqueness is not guaranteed and generally depends on the underlying model parameters. When multiple minimizers exist, they are equally efficient with respect to the total variance criterion. In such cases, any of them may be selected; in practice, the choice can be guided by the external forcing of most interest, such as the anthropogenic forcing (ANT). The optimization can be carried out via a grid search over the parameter space. For the recommended choice of $k=3$, the procedure is computationally manageable. Detailed implementation guidelines are provided in Section~\ref{appendix:sec:implement_details} of the Appendix.

The proposed strategy achieves asymptotic optimality with respect to the total variance of the TLS estimator, within the class of polynomial shrinkage estimators for the covariance matrix. By allowing flexible polynomial forms, the procedure accommodates a wide range of regularization behaviors beyond simple linear shrinkage. This flexibility enables the method to adapt more closely to the spectral structure of the population covariance matrix, thereby improving efficiency while maintaining numerical stability in high-dimensional settings.

\subsection{Estimation of Variability Inflation Factors} \label{subsec:estimation_a}

In this section, we present consistent estimators for the variability inflation factors under two cases: (i) when \(\sigma_Z^2\) is specified \emph{a priori} and only \(\sigma_X^2\) is estimated, and 
(ii) when both \(\sigma_X^2\) and \(\sigma_Z^2\) are unknown and estimated from the data.

Define the scaled within-ensemble total variation in the observed fingerprints as 
\begin{equation}\label{eq:def_VX}
V_X = \frac{1}{N(n_T - p)} \sum_{i=1}^p \sum_{j=1}^{n_i} (\tilde{X}_{ij} - \tilde{X}_i)^T (\tilde{X}_{ij} - \tilde{X}_i),
\end{equation}
where $n_T = \sum_{i=1}^p n_i$ is the total ensemble size. Define the scaled total variation in the control runs as 
\begin{equation}\label{eq:def_VZ}
V_Z = \frac{1}{Nm}\sum_{j=1}^m Z_j^T Z_j = \frac{1}{N} \tr(S).
\end{equation}

\begin{theorem}
    \label{thm:asymptotic_VX_VZ}
    Suppose \ref{enum:HD_regime}--\ref{enum:converge_X_Sigma_1} hold. Let $M_1 = N^{-1}\tr(\Sigma)$ and $M_2^2 = 2N^{-1}\tr(\Sigma^2)$.  As $N \to \infty$, we have 
    \[ \frac{\sqrt{N(n_T-p)}}{\sigma_X^2 M_2 }  \Big(V_X - \sigma^2_X M_1\Big)\stackrel{D}{\longrightarrow} \calN(0,1),\]
    \[\frac{\sqrt{Nm}}{\sigma_Z^2M_2} \Big(V_Z - \sigma^2_Z M_1\Big)\stackrel{D}{\longrightarrow} \calN(0,1). \]
\end{theorem}


\subsubsection{The Case When $\sigma_Z^2$ is Known}\label{subsubsec:when_sigmaZ_known}

If $\sigma_Z^2$ is known as \emph{a priori}, we propose to estimate $\sigma_X^2$ by 
\begin{equation}\label{eq:def_tilde_sigma_X}
\widetilde{\sigma}^2_X = \sigma_Z^2 V_X/ V_Z. 
\end{equation}
The asymptotic properties of the estimator are summarized in the following corollary.
\begin{corollary}
    \label{corollary:consistency_sigma_hat_Z_known}
    Suppose \ref{enum:HD_regime}--\ref{enum:converge_X_Sigma_1} hold. Define 
    \[ \hat{H}^2 = 2 \Big[ \frac{1}{N}\tr(S^2) - V_Z^2 \Big].\]
    Then, we have $\hat{H}^2  - \sigma_Z^4 M_2^2 \stackrel{P}{\longrightarrow}0$. It follows that 
    \[ \sqrt{N(n_T - p)} \frac{\widetilde{\sigma}^2_X/\sigma^2_X - 1}{ \hat{H}  /V_Z} \stackrel{D}{\longrightarrow}\calN(0,1),  \quad \mbox{as } N\to\infty.\]
\end{corollary}
The asymptotic behavior of the TLS estimator with an estimated inflation factor is characterized as follows. 
\begin{theorem}
    \label{thm:asymptotic_normality_estimated_sigma_Z_known}
    Suppose that \ref{enum:HD_regime}--\ref{enum:converge_X_Sigma_1} hold. Let $f$ be any admissible polynomial.  If $\sigma_Z^2$ is known and $\sigma_X^2$ is estimated by $\widetilde{\sigma}_X^2$ as in \eqref{eq:def_tilde_sigma_X}, the TLS estimator $ \hat{\beta}(\widetilde{\sigma}_X^2, f) =  \hat{\beta}(\widetilde{\sigma}_X^2, f(S))$ is such that 
    \[ N^{1/2} \Big[\Xi(f) + \Gamma(f)\Big]^{-1/2} \Big(\hat{\beta}(\widetilde{\sigma}_X^2, f)  - \beta\Big) \stackrel{D}{\longrightarrow} \calN(0, I_p),\]
    where $\Xi(f)$ is defined in Theorem \ref{thm:normality} and 
    \[ \Gamma(f) = \frac{\sigma_X^4  \Pi^2(f) (M_2/M_1)^2}{(n_T - p)} \Delta^{-1}(f)D\beta \beta^T D \Delta^{-1}(f).\]
\end{theorem}

\begin{remark}
    \label{remark:discussion_finite_D}
    Analogous to Remark~\ref{remark:normality_when_ensemble_size_bounded}, 
    Theorem~\ref{thm:asymptotic_normality_estimated_sigma_Z_known} remains valid even when the ensemble sizes \(n_j\) are bounded (\(n_j < \infty\) for \(j = 1, \dots, p\)) as \(N \to \infty\). 
\end{remark}
\begin{remark}
\label{remarK:discussion_infinite_D}
Under \ref{enum:HD_regime}, as $\min_{1\leq i\leq p} n_i \to \infty$, we have $\|\Gamma(f)\|_2 \stackrel{P}{\longrightarrow}0$. By Slutsky's theorem, the conclusion of Theorem~\ref{thm:asymptotic_normality_estimated_sigma_Z_known} remains valid when $\Xi(f)+\Gamma(f)$ is replaced by $\Xi(f)$. In this sense, the effect of estimating the inflation factor on the asymptotic distribution of the TLS estimator is negligible. However, in finite samples, incorporating \(\Gamma(f)\) generally improves the approximation of the variance of $\hat{\beta}(\widetilde{\sigma}_X^2, f)$.
\end{remark}

When the inflation factor $\sigma_X^2$ is estimated by $\widetilde{\sigma}_X^2$, the shrinkage selection procedure described in Section \ref{subsec:selection_shrinkage} can be adapted accordingly. Specifically, define 
\[
\hat\Gamma(f, \widetilde{\sigma}_X^2) 
= 
\frac{\widetilde{\sigma}_X^4 \Pi^2(f)\hat{H}^2/V_Z^2}
{n_T - p} 
\;\Delta^{-1}(f) D\hat{\beta}(f) \hat{\beta}^\top(f) D \Delta^{-1}(f).
\]
Moreover, let \(\hat{\Xi}(f,\widetilde{\sigma}_X^2) = \hat{\Xi}(f, \widetilde{\sigma}_X^2, \sigma_Z^2)\) as in \eqref{eq:def_hat_xi} but using the estimated factor $\widetilde{\sigma}_X^2$. We then propose selecting the shrinkage function \(f\) by minimizing the estimated total variance
\[
\hat{f} 
= 
\arg\min_{f \in \frakF} 
\operatorname{trace}\Big(
\hat{\Xi}(f, \widetilde{\sigma}_X^2) + \hat{\Gamma}(f, \widetilde{\sigma}_X^2)
\Big).
\]
The implementation details are analogous to those described in Section~\ref{subsec:selection_shrinkage} and are therefore omitted.

In Algorithm~\ref{algo:estimation_1}, we present a summary of the proposed procedure for jointly estimating $\beta$ and $\sigma_X^2$ under the assumption that $\sigma_Z^2$ is known.
\begin{algorithm}[h]
\caption{Estimation of $\beta$ and $\sigma^2_X$ when $\sigma^2_Z$ is known}
\label{algo:estimation_1}
\begin{algorithmic}
    \STATE \textbf{Input:} Response $Y$, fingerprint ensembles $\{\tilde{X}_{ij}\}$, control runs $\{Z_j\}$, inflation factor $\sigma_Z^2$, and shrinkage family $\frakF$.
    \STATE \textbf{Step 1:} Compute $V_X$ and $V_Z$ as defined in \eqref{eq:def_VX} and \eqref{eq:def_VZ}.
    \STATE \textbf{Step 2:} Estimate the inflation factor $\sigma_X^2$ by 
        $\widetilde{\sigma}_X^2 = \sigma_Z^2 {V_X}/{V_Z}$.
    \STATE \textbf{Step 3:} For each $f \in \frakF$, compute $\hat{\Xi}(f, \widetilde{\sigma}_X^2) + \hat{\Gamma}(f, \widetilde{\sigma}_X^2)$, and select the data-driven optimal shrinkage function
        $\hat{f} = \arg\min_{f\in\frakF} \operatorname{trace}\!\big(\hat{\Xi}(f, \widetilde{\sigma}_X^2) + \hat{\Gamma}(f, \widetilde{\sigma}_X^2)\big)$.
    \STATE \textbf{Step 4:} Compute the TLS estimator 
        $\hat{\beta} = \hat{\beta}(\widetilde{\sigma}_X^2,\hat{f}(S))$,
    and obtain the associated quantities $\Delta$, $\Pi$, $\Omega$, $K$, and $\hat{\Xi}+\hat{\Gamma}$ for $f=\hat{f}$.
    \STATE \textbf{Output:} TLS estimator $\hat{\beta}$, estimated inflation factor $\widetilde{\sigma}_X^2$, optimal shrinkage $\hat{f}$, and associated objects $\Delta$, $\Pi$, $\Omega$, $K$, and $\hat{\Xi}+\hat{\Gamma}$.
\end{algorithmic}
\end{algorithm}

\subsubsection{The Case When $\sigma_Z^2$ is Unknown}\label{subsubsec:when_sigmaZ_unknown}

A challenging case arises when both $\sigma_X^2$ and $\sigma_Z^2$ must be estimated in a data-driven manner. We begin by assuming that a regularized estimator $f(S)$ of the covariance matrix is given, where $f$ can be any admissible polynomial. A candidate choice is the linear shrinkage estimator $f(S) = wS + I_N$ by \citet{ledoit2004well} after normalization. For completeness, the expression for $w$ is provided in Section \ref{appendix:sec:implement_details} of the Appendix. 

Consider the TLS estimator $\hat{\beta}(\sigma^2) = \hat{\beta}(\sigma^2, f(S))$, where $\sigma^2$ is restricted to the interval $[0, \Upsilon]$. Here, $\Upsilon>0$ is a sufficiently large upper bound. As established in Theorem \ref{thm:consistency_a_b}, the choice of $\Upsilon$ does not materially affect the procedure, provided that the interval contains the true factor $\sigma_X^2$. Practical guidance for selecting $\Upsilon$ is given in Section \ref{appendix:sec:implement_details} of the Appendix.

Define the loss function 
\begin{equation}
        \label{eq:loss_estimate_a}
        {L}(\sigma^2) = \left|\frac{\sigma^2\| Y- \tilde{X} \hat{\beta}(\sigma^2) \|_2^2}{ N [1+ \sigma^2\hat{\beta}^T(\sigma^2) D \hat{\beta}(\sigma^2)] } -  V_X \right|. 
    \end{equation}
We estimate $\sigma^2_X$ and $\sigma^2_Z$ by 
\begin{equation}\label{eq:def_widehat_sigma}
\left\{
\begin{aligned}
&\widehat{\sigma}^2_X =\arg\min_{0\leq \sigma^2\leq \Upsilon} L(\sigma^2),\\
&\widehat{\sigma}^2_Z  =  \widehat{\sigma}^2_X  V_Z/V_X.
\end{aligned}
\right.
\end{equation} 
\begin{theorem}
    \label{thm:consistency_a_b}
  Suppose \ref{enum:HD_regime}--\ref{enum:converge_X_Sigma_1} hold, and let $f$ be any admissible polynomial. Assume that $\Upsilon > \sigma^2_X$. Then, as $N\to\infty$, we have 
  \[ \widehat\sigma^2_X = \sigma^2_X + O_p(N^{-1/2}) \quad \mbox{ and }\quad \widehat{\sigma}_Z^2 = \sigma^2_Z + O_p(N^{-1/2}).\]
  Consider the TLS estimator $\hat{\beta}(\widehat{\sigma}^2_X) = \hat{\beta}(\widehat{\sigma}_X^2, f(S))$ with the estimated inflation factor $\widehat{\sigma}_X^2$. We have 
    \[ \hat{\beta}(\widehat{\sigma}^2_X) - \hat{\beta}(\sigma_X^2) = O_p\Big(\frac{1}{n_T \sqrt{N}}\Big), \]
    where $n_T = \sum_{i=1}^p n_i$ is the total ensemble size. It follows then
    \begin{equation}\label{eq:normality_beta_with_estimated_sigma}
    \sqrt{N} [\Xi(f)]^{-1/2} \Big( \hat{\beta}(\widehat{\sigma}_X^2) - \beta\Big) \stackrel{D}{\longrightarrow} \calN(0, I_p).
    \end{equation}
\end{theorem}

\begin{remark}
\label{remark:double_dipping}
The asymptotic normality result in \eqref{eq:normality_beta_with_estimated_sigma} is established under the asymptotic regime where \(n_T \to \infty\), as specified in \ref{enum:HD_regime}. From a practical standpoint, however, it is important to evaluate the finite-sample performance of \(\hat{\beta}(\widehat{\sigma}_X^2)\), particularly in situations where \(n_T\) is relatively small. 
To this end, Section~\ref{sec:numerical} presents a numerical study evaluating the performance of the estimator in the setting \(n_j = 40\). We observe that results using estimated and true values of $\sigma^2_X$ and $\sigma^2_Z$ are comparable, supporting the reliability of the variance estimation procedure.
\end{remark}

In Algorithm \ref{algo:estimation}, we present a summary of the proposed procedure for simultaneously estimating $\beta$, $\sigma_X^2$, and $\sigma_Z^2$. 
\begin{algorithm}[h]
\caption{Estimation of $\beta$, $\sigma^2_X$ and $\sigma^2_Z$}
\label{algo:estimation}
\begin{algorithmic}
    \STATE \textbf{Input:} Response $Y$, fingerprint ensemble $\tilde{X}_{ij}$, control runs $Z_j$'s, shrinkage family $\frakF$, upper bound $\Upsilon$.
    \STATE \textbf{Step 1:} Compute $V_X$ and $V_Z$ as defined in \eqref{eq:def_VX} and \eqref{eq:def_VZ}, respectively.
    \STATE \textbf{Step 2:} Set $W = w S+ I_N $, where $w$ is given in Lemma \ref{lemma:ledoit_wolf_w}. For each $\sigma^2\in [0, \Upsilon]$, compute $\hat{\beta}(\sigma^2) = \hat{\beta}(\sigma^2, W)$, using the closed-form expression given in Lemma \ref{lemma:explicit_formula_of_beta}.
    \STATE \textbf{Step 3:} Compute the loss function $L(\sigma^2)$ as in \eqref{eq:loss_estimate_a}, based on $\hat{\beta}(\sigma^2)$. Obtain the estimates of the inflation factors: $\widehat{\sigma}^2_X =\arg\min_{\sigma^2\leq \Upsilon} L(\sigma^2)$ with a grid search and $\widehat{\sigma}^2_Z = \widehat{\sigma}^2_XV_Z/V_X$.
    \STATE \textbf{Step 4:} Obtain the estimate of the optimal shrinkage $\hat{f} = \arg \min_{f\in\frakF} \operatorname{trace}[\hat{\Xi}(f)]$, where $\hat{\Xi}(f)$ is computed using  \eqref{eq:partial_fraction1} -- \eqref{eq:def_hat_xi} with the estimated factors $\widehat{\sigma}^2_X$ and $\widehat{\sigma}^2_Z$.
    \STATE \textbf{Step 5:} Repeat Step 2-4 with $W = \hat{f} (S)$ to refine the estimates. 
    \STATE \textbf{Step 6}: Compute the final TLS estimator of $\beta$ as  $\hat{\beta} = \hat{\beta}(\widehat{\sigma}^2_X, \hat{f}(S))$ and all associated objects $\Delta$, $\Pi$, $\Omega$, $K$ $\hat{\Xi}$ for $f = \hat{f}$.
    \STATE \textbf{Output:}  TLS Estimator $\hat{\beta}$, estimated factors $\widehat{\sigma}^2_X$ and $\widehat{\sigma}^2_Z$, optimal shrinkage $\hat{f}$, associated objects $\Delta$, $\Pi$, $\Omega$, $K$ $\hat{\Xi}$.
\end{algorithmic}
\end{algorithm}

\subsection{Related Inferences and Residual Diagnostics}\label{subsec:CI_Diagnosis}
In this section, we present methods for related inference and residual diagnostics, focusing on the case where both $\sigma_X^2$ and $\sigma_Z^2$ are unknown. We use the estimates $\hat{\beta}$, $\widehat{\sigma}_X^2$, $\widehat{\sigma}_Z^2$, $\hat{f}$, and $\hat{\Xi}$ from Algorithm~\ref{algo:estimation}. When $\sigma_Z^2$ is known, the same procedures apply with estimates from Algorithm~\ref{algo:estimation_1}.

\subsubsection{Hypothesis Testing}
We consider testing hypotheses about the scaling factor $\beta$ of the form
\[
H_0: B\beta = b 
\qquad \text{versus} \qquad 
H_a: B\beta \neq b,
\]
where $B$ is  a $ p_0 \times p$ constraint matrix of rank $p_0\leq p$, and $b$ is the specified vector under the null.  
Motivated by Theorem \ref{thm:consistency_a_b}, we propose a Wald-type test, which rejects the null hypothesis at asymptotic level $\alpha$ when
\[
T \coloneqq N\big(B\hat{\beta} - b\big)^\top [B\hat{\Xi}B^T]^{-1} \big(B\hat{\beta} - b\big) > \chi^2_{p_0,\,1-\alpha},
\]
where $\chi^2_{p_0,\,1-\alpha}$ is the $(1-\alpha)$ quantile of the chi-square distribution with $p_0$ degrees of freedom. 

\subsubsection{Confidence Intervals}
Accurate and reliable confidence intervals for the scaling factors $\beta_i$ play a critical role in attribution analysis of climate variables. For each scaling factor $\beta_i$, $i=1,\dots, p$, a marginal confidence interval at the asymptotic confidence level  \( (1 - \alpha) \) can be constructed as 
\[
\beta_i \in \left[ \hat{\beta}_i - \frac{z_{1 - \alpha/2}}{\sqrt{N}}  \hat{\Xi}_{ii}^{1/2} , \, \, \hat{\beta}_i + \frac{z_{1 - \alpha/2}}{\sqrt{N}}  \hat{\Xi}_{ii}^{1/2}  \right],
\]
where \( z_{1 - \alpha/2} \) denotes the upper \( (1 - \alpha/2) \)-quantile of the standard normal distribution, $\hat{\beta}_i$  is the $i$th element of $\hat{\beta}$, and \( \widehat{\Xi}_{ii}\) is the \( i \)th diagonal element of $\hat{\Xi}$. To control the family-wise error rate across the marginal intervals, the Bonferroni correction can be applied by adjusting the individual confidence level to $(1-\alpha/p)$.

A joint confidence region for $\beta$ at the asymptotic level $(1-\alpha)$ can be constructed as 
\[
\left\{ \beta \in \mathbb{R}^p ~:~  N(\beta-\hat{\beta})^T \hat{\Xi}^{-1} (\beta-\hat{\beta}) \leq \chi^2_{p,\,1 - \alpha} \right\}.
\]


\subsubsection{Residual Diagnostics}
Assessing model adequacy is a critical aspect of statistical modeling, yet it has received limited attention in the context of fingerprinting. The problem is particularly challenging because of the strong correlation structure of the residuals and the lack of a consistent estimator for the covariance matrix. To address this challenge, we propose a residual-based diagnostic procedure. Specifically, define the residuals of the regression model as 
\[ \hat{\epsilon} = Y - \tilde{X}\hat{\beta}.\]

As an initial step, informal diagnostic plots of the residuals can provide useful insights into potential departures from the model specification. To this end, we standardize the residuals to have approximate unit variance by
\[\hat{\epsilon}_{\rm std} =  [\operatorname{Diag}(S/\widehat{\sigma}^2_Z)]^{-1/2} \hat{\epsilon}.\] 
Plotting these standardized residuals against their indices, spatial or temporal coordinates, or the observed responses $Y$ may reveal systematic patterns. Substantial deviation from a pattern centered around zero with constant variance may suggest violations of model assumptions or evidence of model misspecification.

Secondly, we develop a formal residual consistency test. Let 
\[ Q(\hat{\epsilon}) = \frac{1}{N}\hat{\epsilon}^T[\hat{f}(S)]^{-1} \hat{\epsilon} - [1+ \widehat{\sigma}_X^2\hat{\beta}^T D\hat{\beta}]\Pi(\hat{f}).\]
 \begin{theorem}
     \label{thm:adequacy_test}
     Suppose \ref{enum:HD_regime}--\ref{enum:converge_X_Sigma_1} hold. Under the model assumptions \eqref{eq:fingerprint1}--\eqref{eq:fingerprint3}, we have 
     \[ \frac{Q(\hat{\epsilon})}{\sqrt{V_Q}} \stackrel{D}{\longrightarrow} \calN(0,1), \quad \mbox{where}\quad  V_Q  = \frac{2}{N} (1+ \widehat{\sigma}_X^2 \hat{\beta}^T D \hat{\beta})^2 K(\hat{f}).\]
 \end{theorem}
From Theorem~\ref{thm:adequacy_test}, the model~\eqref{eq:fingerprint1}--\eqref{eq:fingerprint3} is rejected as misspecified at asymptotic significance level \(1 - \alpha\), if
\[
\frac{Q(\hat{\epsilon})}{\sqrt{V_Q}} > z_{1-\alpha},
\]
where \(z_{1-\alpha}\) denotes the \(1 - \alpha\) quantile of $\calN(0,1)$.


\section{Simulation study}
\label{sec:numerical}

To assess the finite-sample performance of the proposed method relative to existing approaches, we conducted simulation experiments designed to mimic detection and attribution analyses of annual mean near-surface air temperature anomalies in the North American region ($130^\circ\text{W}$–$50^\circ\text{W}$, $30^\circ\text{N}$–$60^\circ\text{N}$). We focused on anthropogenic (ANT) and natural (NAT) forcings, the primary external drivers of temperature variability in attribution studies. To reflect realistic climate signals, the underlying fingerprints $X_{\text{ANT}}$ and $X_{\text{NAT}}$ were derived from the CMIP6 multi-model ensemble mean \citep{eyring2016overview} based on 40 ensemble runs. The region was represented on a $10^\circ \times 5^\circ$ latitude--longitude grid, yielding 48 spatial boxes, and five-year averages were taken over 1951--2020, producing 13 temporal points. After removing missing values, each fingerprint had dimension $N = 624$. 

In each simulation setting, the true scaling factors were set to $\beta_{\text{ANT}} = \beta_{\text{NAT}} = 1$. The observed mean temperature vector $Y$, the fingerprints $\tilde X_{\text{ANT}}$ and $\tilde X_{\text{NAT}}$, and the independent control runs $Z$ were generated according to Models~\eqref{eq:fingerprint1}--\eqref{eq:fingerprint3} from multivariate normal distributions with covariance matrix $\Sigma$ and scale inflation factors $\sigma_X^2$ and $\sigma_Z^2$. We considered sample sizes $m \in \{200, 400, 800\}$, corresponding to the number of independent realizations of $Z$, and ensemble sizes $n_1 = n_2 = 40$, corresponding to the number of runs used to construct $\tilde X_{\text{ANT}}$ and $\tilde X_{\text{NAT}}$. To vary signal strength, we considered two levels with $\sigma^2_X = \sigma^2_Z \in \{1, 2\}$.

The underlying error covariance matrix $\Sigma$ was constructed from 181 preindustrial control runs of CMIP6 simulations. From these runs we obtained the sample covariance matrix $S$ and considered two shrinkage-based structures for $\Sigma$. The first, denoted $\Sigma_{LS}$, used a linear shrinkage form that favors estimators built on linear shrinkage. The second, denoted $\Sigma_{PL}$, applied an $S$-shaped cubic polynomial to the eigenvalues of $S$ over the range from zero to the maximum eigenvalue, designed to mimic eigenvalue distributions frequently observed in climate model control runs. These two scenarios allow us to evaluate robustness when the true covariance either conforms to a simple linear shrinkage assumption or follows a more flexible nonlinear structure. For each configuration, 500 replicates were generated.

We compared the TLS estimators in \eqref{eq:def_TLS} under two approaches for estimating the weight matrix: (i) a linear shrinkage estimator, denoted by ``LS'', which combines the sample covariance with the identity matrix, first introduced into climate change studies by \citet{ribes2009adaptation} and later refined by \citet{li2025regularized} to achieve optimality within the class of linear shrinkage, and (ii) the proposed polynomial shrinkage estimator in \eqref{eq:poly}, denoted by ``POLY'', with degree $k=3$ and coefficients $(\ell_3, \ell_2, \ell_1)$ selected by grid search over $\ell_i \in (-10, 10), \; i=1,2,3$, subject to monotonicity constraints. The variance components $\sigma_X^2$ and $\sigma_Z^2$ were estimated using the methods in Section~\ref{subsec:estimation_a}. As a benchmark, we also considered TLS with prewhitening based on the inverse of the true covariance matrix. Confidence intervals were constructed from the asymptotic normality of the TLS estimators with consistent covariance estimation, as described in Section~\ref{subsec:CI_Diagnosis}.

Here we report results for the ANT forcing, while full numerical results for both ANT and NAT forcings are provided in the Supplementary Material and show similar patterns. Table~\ref{tab:bias_rmse_b1} summarizes the empirical bias and root mean square errors (RMSEs) for estimating $\beta_{\text{ANT}}$, and Table~\ref{tab:coverage_length_b1} presents the empirical coverage rates and average lengths of the 90\% confidence intervals obtained from the competing methods.

\begin{table}[htbp]
\centering
\caption{Empirical bias (RMSE) of $\hat\beta_{\text{ANT}}$ for the ANT forcing using the TLS estimator in \eqref{eq:def_TLS} with linear and polynomial shrinkage weight matrices under both estimated and true variances $(\sigma_x^2,\sigma_z^2)$. The ensemble sizes are $n_1 = n_2 = 40$. Results based on prewhitening using $\Sigma^{-1}$, the inverse of the true covariance matrix are also reported as a benchmark.}
\label{tab:bias_rmse_b1}
\begin{tabular}{cccccc}
\toprule
& \multicolumn{2}{c}{Estimated $\sigma^2_x$ and $\sigma^2_z$} & \multicolumn{3}{c}{True $\sigma^2_x$ and $\sigma^2_z$} \\
\cmidrule(lr){2-3} \cmidrule(lr){4-6}
$m$ & LS & POLY & LS & POLY & True $\Sigma^{-1}$\\
\midrule
\multicolumn{6}{c}{Setting 1: $\Sigma_{LS}$, $\sigma^2_x = \sigma^2_z = 1$} \\
200 & 0.002 (0.077) & 0.002 (0.078) & -0.000 (0.076) & -0.000 (0.077) & \\
400 & 0.003 (0.068) & 0.003 (0.070) & ~0.000 (0.067) & ~0.000 (0.068) & \\
800 & 0.002 (0.063) & 0.003 (0.062) & -0.000 (0.061) & -0.000 (0.061) & -0.001 (0.046) \\
\midrule
\multicolumn{6}{c}{Setting 2: $\Sigma_{LS}$, $\sigma^2_x = \sigma^2_z = 2$} \\
200 & 0.002 (0.081) & 0.004 (0.083) & -0.003 (0.080) & -0.002 (0.081) & \\
400 & 0.004 (0.069) & 0.005 (0.073) & -0.001 (0.068) & ~0.000 (0.072) & \\
800 & 0.005 (0.064) & 0.007 (0.067) & ~0.001 (0.062) & ~0.001 (0.065) & -0.002 (0.053) \\
\midrule
\multicolumn{6}{c}{Setting 3: $\Sigma_{PL}$, $\sigma^2_x = \sigma^2_z = 1$} \\
200 & 0.003 (0.079) & ~0.000 (0.058) & ~0.002 (0.079) & -0.000 (0.059) &\\
400 & 0.000 (0.050) & -0.001 (0.029) & -0.000 (0.049) & -0.001 (0.031) &\\
800 & 0.001 (0.032) & -0.000 (0.023) & ~0.000 (0.032) & -0.001 (0.024) & -0.001 (0.016) \\
\midrule
\multicolumn{6}{c}{Setting 4: $\Sigma_{PL}$, $\sigma^2_x = \sigma^2_z = 2$} \\
200 & 0.002 (0.096) & 0.001 (0.060) & -0.000 (0.093) & -0.001 (0.061) & \\
400 & 0.000 (0.065) & 0.000 (0.029) & -0.002 (0.063) & -0.001 (0.030) & \\
800 & 0.000 (0.045) & 0.000 (0.022) & -0.001 (0.043) & -0.001 (0.023) & -0.001 (0.018) \\
\bottomrule
\end{tabular}
\end{table}

\begin{table}[htbp]
\centering
\caption{Empirical coverage (length) of 90\% confidence intervals for $\hat\beta_{\text{ANT}}$ for the ANT forcing using the TLS estimator in \eqref{eq:def_TLS} with linear and polynomial shrinkage weight matrices under both estimated and true variances $(\sigma_x^2,\sigma_z^2)$. The ensemble sizes are $n_1 = n_2 = 40$.}
\label{tab:coverage_length_b1}
\begin{tabular}{ccccc}
\toprule
& \multicolumn{2}{c}{Estimated $\sigma^2_x$ and $\sigma^2_z$} & \multicolumn{2}{c}{True $\sigma^2_x$ and $\sigma^2_z$} \\
\cmidrule(lr){2-3} \cmidrule(lr){4-5}
$m$ & LS & POLY & LS & POLY \\
\midrule
\multicolumn{5}{c}{Setting 1: $\Sigma_{LS}$, $\sigma^2_x = \sigma^2_z = 1$} \\
200 & 87.6 (0.234) & 85.6 (0.228) & 88.4 (0.238) & 87.0 (0.232) \\
400 & 87.4 (0.209) & 85.8 (0.213) & 88.2 (0.212) & 88.0 (0.216) \\
800 & 88.8 (0.199) & 88.2 (0.195) & 88.6 (0.202) & 89.0 (0.198) \\ 
\midrule
\multicolumn{5}{c}{Setting 2: $\Sigma_{LS}$, $\sigma^2_x = \sigma^2_z = 2$} \\
200 & 87.0 (0.250) & 85.2 (0.241) & 88.0 (0.252) & 85.4 (0.243) \\
400 & 88.8 (0.217) & 88.4 (0.227) & 90.0 (0.217) & 90.2 (0.228) \\
800 & 86.2 (0.198) & 86.6 (0.206) & 89.4 (0.199) & 88.2 (0.207) \\ 
\midrule
\multicolumn{5}{c}{Setting 3: $\Sigma_{PL}$, $\sigma^2_x = \sigma^2_z = 1$} \\
200 & 89.2 (0.267) & 88.6 (0.192) & 90.6 (0.269) & 90.8 (0.197) \\
400 & 90.0 (0.167) & 91.2 (0.098) & 92.2 (0.168) & 91.2 (0.104) \\
800 & 89.0 (0.106) & 88.6 (0.073) & 90.2 (0.106) & 90.4 (0.078) \\
\midrule
\multicolumn{5}{c}{Setting 4: $\Sigma_{PL}$, $\sigma^2_x = \sigma^2_z = 2$} \\
200 & 89.2 (0.312) & 89.6 (0.200) & 90.2 (0.312) & 91.8 (0.205) \\
400 & 90.6 (0.215) & 90.8 (0.101) & 91.4 (0.215) & 92.0 (0.108) \\
800 & 90.8 (0.145) & 91.4 (0.076) & 91.6 (0.145) & 93.0 (0.081) \\ 
\bottomrule
\end{tabular}
\end{table}

From Table~\ref{tab:bias_rmse_b1}, all methods recover the true scaling factor $\beta_{\text{ANT}}$ well on average. The RMSEs decrease as the sample size $m$ increases, confirming the statistical consistency of the estimators, and they also decline as the signal strength increases (i.e., smaller $\sigma^2_X$ and $\sigma^2_Z$), indicating that estimation improves with stronger signals. The polynomial shrinkage method appears less sensitive to changes in signal strength. Results using estimated and true values of $\sigma^2_X$ and $\sigma^2_Z$ are comparable, supporting the validity of the estimation procedure for the inflation factors described in Section~\ref{subsec:estimation_a}. As expected, the TLS estimator with the true $\Sigma^{-1}$ achieves the smallest RMSEs in all cases. When the true covariance follows a linear shrinkage structure ($\Sigma_{LS}$), the LS and POLY methods yield similar performance, and as $m$ increases, their RMSEs approach those under the true covariance. By contrast, when the true covariance follows a polynomial structure ($\Sigma_{PL}$), the proposed POLY method shows a clear advantage over LS and attains results comparable to the truth when the sample size $m$ is of the same order as the dimension $N$.

Confidence intervals are a central tool for detection and attribution analyses, and it is desirable that the asymptotic variance in Section~\ref{subsec:selection_shrinkage} yields valid inference for the scaling factors. From Table~\ref{tab:coverage_length_b1}, the coverage rates across all methods remain close to the nominal 90\% level, with interval lengths shrinking as the sample size $m$ increases, reflecting the consistency of the variance estimation. In settings where the true covariance follows $\Sigma_{LS}$, the LS and POLY methods yield intervals of similar length. In contrast, when the true covariance follows $\Sigma_{PL}$, the proposed POLY method produces substantially shorter intervals while maintaining coverage accuracy, demonstrating a marked efficiency gain. Results using estimated and true values of $\sigma^2_X$ and $\sigma^2_Z$ are again comparable, further supporting the reliability of the variance estimation procedure. Overall, the polynomial shrinkage method delivers inference as reliable as linear shrinkage in terms of coverage, while offering notable improvements in efficiency, particularly under nonlinear covariance structures.

We also conducted a diagnostic analysis of the inflation factors $\sigma_X^2$ and $\sigma_Z^2$, which capture differences in the magnitude of climate variability between observations and model simulations. Figure~\ref{fig:inflation_densities} shows density plots of their sampling distributions. The empirical distributions are approximately chi-square, as expected, and the averages over 500 replicates (vertical lines) lie close to the true values, providing further validation of the estimation procedure.

\begin{figure}[htbp]
\centering
\includegraphics[width=0.9\textwidth]{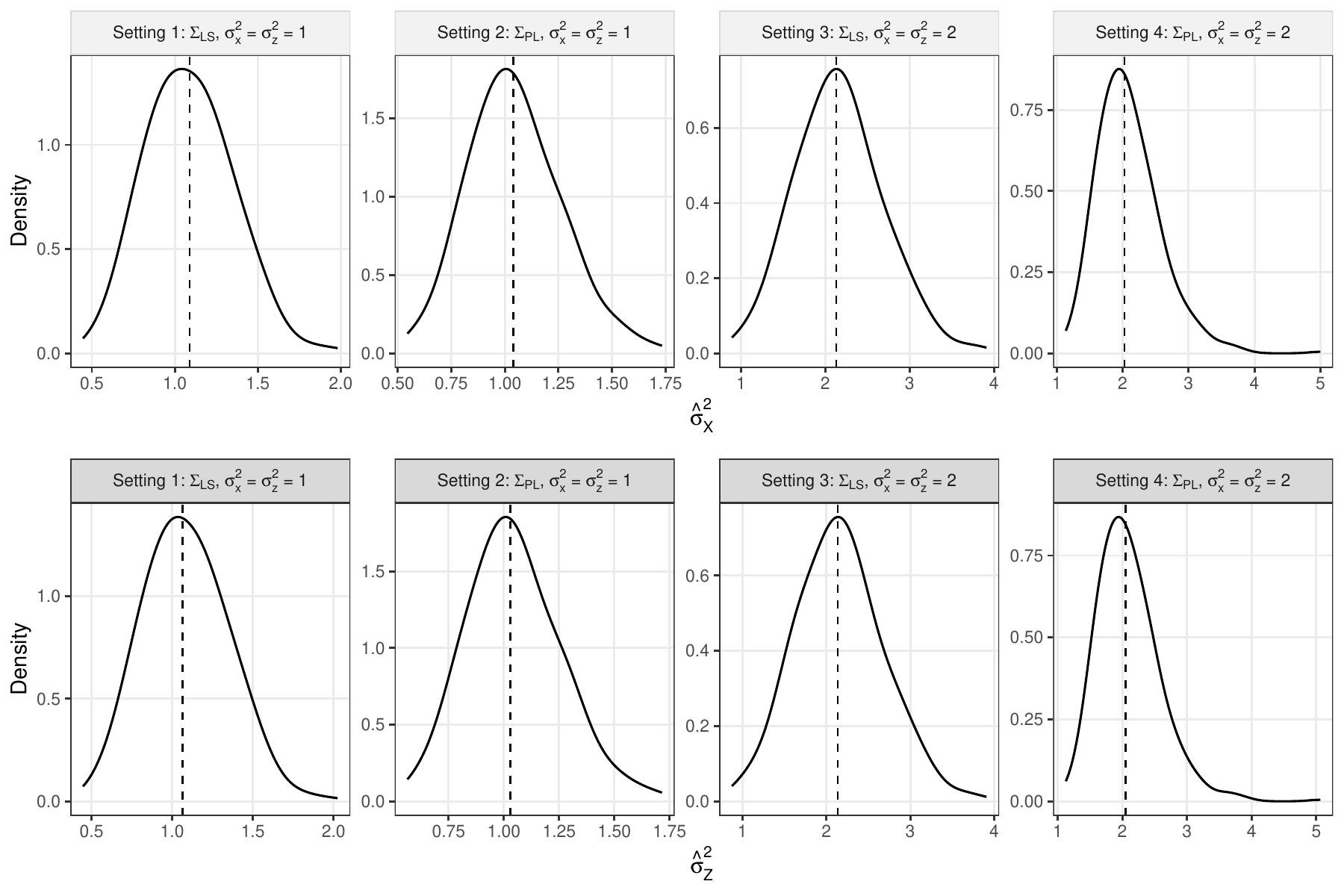}
\caption{Sampling distributions of the estimated inflation factors $\sigma_X^2$ and $\sigma_Z^2$. The vertical lines mark the averages over 500 replicates, which are close to the true values.}
\label{fig:inflation_densities}
\end{figure}

Finally, we assessed the residual consistency test introduced in Section~\ref{sec:methodology}, which evaluates whether the assumed linear model is correctly specified. The test examines the adequacy of key assumptions, including the linear relationship between fingerprints and the observed response, the correctness of the assumed covariance structure, and the validity of the heterogeneous error model. We conducted simulations under the four null settings described above to assess its size. For alternative scenarios, we introduced heterogeneous inflation factors for $\sigma_X^2$ and $\sigma_Z^2$ to induce model misspecification, with variance inflation increasing over the time period considered in the detection analysis, reflecting cases in which climate model simulations exhibit progressively larger variability relative to the observations.

We report results based on the true values of $\sigma_X^2$ and $\sigma_Z^2$, since the estimated values exhibit high variability, as illustrated by the density plots in Figure~\ref{fig:inflation_densities}. This variability arises from the limited size of the ensemble runs, which makes accurate estimation of the variance inflation factors difficult in practice. When $\sigma_X^2$ and $\sigma_Z^2$ are estimated, the resulting variance estimates tend to be slightly underestimated, leading to moderate over-rejection of the null hypothesis (around 15\%). The results summarized in Table~\ref{tab:residuals} correspond to the case with the true values of $\sigma_X^2$ and $\sigma_Z^2$. As shown in the table, the test maintained rejection rates close to the nominal 5\% level under the null and exhibited high power under the alternative settings, confirming its effectiveness as a diagnostic for assessing the validity of the assumed linear model in detection and attribution analyses. Developing more robust residual diagnostics that remain stable under estimated variance inflation factors and limited ensemble sizes requires further thorough investigation.

\begin{table}[htbp]
\centering
\caption{Empirical rejection rates (\%) of the residual consistency test under null and alternative (heterogeneous variance inflation), shown side-by-side for LS and POLY. The ensemble sizes are $n_1 = n_2 = 40$.}
\label{tab:residuals}
\begin{tabular}{r r r r r}
\toprule
$m$ & \multicolumn{2}{c}{Null} & \multicolumn{2}{c}{Alternative} \\
\cmidrule(lr){2-3}\cmidrule(lr){4-5}
& LS & POLY & LS & POLY \\
\midrule
\multicolumn{5}{c}{Setting 1: $\Sigma_{LS}$, $\sigma_X^2=\sigma_Z^2=1$} \\
200 & 6.8 & 6.4 & 100.0 & 100.0 \\
400 & 7.8 & 7.6 & 100.0 & 100.0 \\
800 & 6.4 & 7.2 & 100.0 & 100.0 \\
\midrule
\multicolumn{5}{c}{Setting 2: $\Sigma_{LS}$, $\sigma_X^2=\sigma_Z^2=2$} \\
200 & 7.0 & 6.0 & 100.0 & 100.0 \\
400 & 6.8 & 6.4 & 100.0 & 100.0 \\
800 & 5.8 & 6.0 & 100.0 & 100.0 \\
\midrule
\multicolumn{5}{c}{Setting 3: $\Sigma_{PL}$, $\sigma_X^2=\sigma_Z^2=1$} \\
200 & 5.2 & 5.4 & 100.0 & 100.0 \\
400 & 5.6 & 5.6 & 100.0 & 100.0 \\
800 & 6.2 & 7.2 & 100.0 & 96.6 \\
\midrule
\multicolumn{5}{c}{Setting 4: $\Sigma_{PL}$, $\sigma_X^2=\sigma_Z^2=2$} \\
200 & 6.8 & 7.0 & 100.0 & 100.0 \\
400 & 7.0 & 7.0 & 100.0 & 99.0 \\
800 & 6.2 & 6.8 & 100.0 & 93.4 \\
\bottomrule
\end{tabular}
\end{table}

\section{Application}
\label{sec:application}

To demonstrate the effectiveness of the proposed adaptable fingerprinting framework for estimating scaling factors, we conducted detection and attribution analysis on the annual mean near-surface air temperature from 1951 to 2020 over different continental and subcontinent scales. At the continental level, we considered the Northern Hemisphere (NH), the NH midlatitudes (NHM; $30^\circ$–$70^\circ$N), Eurasia (EA), and North America (NA). At the subcontinental level, we focused on Western North America (WNA), Central North America (CNA), and Eastern North America (ENA). For each region, we examined the contributions of two external forcings, anthropogenic (ANT) and natural (NAT).

Observed mean temperatures $Y \in \mathbb{R}^N$ were obtained from the HadCRUT5 dataset \citep{morice2021updated}, which provides monthly near-surface temperature anomalies from January 1850 onward on a $5^\circ \times 5^\circ$ grid, relative to the 1961–1990 reference period. To reduce dimensionality, we aggregated the raw observations into five-year averages and coarser spatial boxes. Table~\ref{tab:regions} summarizes the resulting spatiotemporal dimensions, including the number of grid boxes, the number of time steps, and the total number of observations after accounting for missing values.

\begin{table}[tbp]
  \footnotesize
  \centering
  \caption{\sf Summaries of the names, coordinate ranges, ideal
    spatio-temporal dimensions ($S$ and $T$), and dimension of
    observation after removing missing values of the 5 regions
    analyzed in the study.}
  \label{tab:regions}
  \def\arraystretch{1}\tabcolsep=0.4em
  \begin{tabular}{llcccccc}
    \toprule
    Acronym & Regions & Longitude  & Latitude   & Grid size & $S$ & $T$ & $N$ \\
            &         & ($^\circ$E) & ($^\circ$N) & ($1^\circ \times 1^\circ$) & & &\\
    \midrule
    \multicolumn{8}{c}{Global and Continental Regions}\\
    GL & Global & $-$180 / 180 & $-$90 / 90 & $40 \times 30$ & 54 & 13 & 696 \\
    NH & Northern Hemisphere & $-$180 / 180 & 0 / 90 & $40 \times 30$ & 27 & 13 & 351 \\
    NHM & Northern Hemisphere $30^\circ N$ to $70^\circ N$ & $-$180 / 180 & 30 / 70 & $40 \times 10$ & 36 & 13 & 468 \\
    EA & Eurasia & $-$10 / 180 & 30 / 70 & $10 \times 20$ & 38 & 13 & 494 \\
    NA & North America & $-$130 / $-$50 & 30 / 60 & $10 \times 5$ & 48 & 13 & 624 \\
    \addlinespace[1ex]
    \multicolumn{8}{c}{Subcontinental Regions}\\
    \midrule
    WNA & Western North America & $-$130 / $-$105 & 30 / 60 & $5 \times 5$ & 30 & 13 & 390 \\
    CNA & Central North America & $-$105 / $-$85 & 30 / 50 & $5 \times 5$ & 16 & 13 & 208 \\
    ENA & Eastern North America & $-$85 / $-$50 & 15 / 30 & $5 \times 5$ & 21 & 13 & 273 \\
    \bottomrule
  \end{tabular}
\end{table}

The fingerprints and control runs were derived from CMIP6 multimodel simulations \citep{eyring2016overview}. 
The simulations included 43 runs of the hist-GHG experiment (driven solely by changes in well-mixed greenhouse gas concentrations), 43 runs of the hist-aer experiment (driven by changes in anthropogenic aerosols), 
40 runs of the hist-nat experiment (driven by natural forcings), and preindustrial control simulations 
of varying lengths representing internal climate variability.

For the NAT forcing, the fingerprint $\tilde{X}_{\text{NAT}}$ was obtained directly by averaging over 
the 40 available hist-nat runs. For the ANT forcing, which is typically the main focus of 
climate detection and attribution studies, direct model outputs were not available in CMIP6. Following the 
standard linear additivity assumption $X_{\text{ANT}} = X_{\text{GHG}} + X_{\text{AER}}$ \citep{zhang2006multimodel}, 
we constructed $\tilde{X}_{\text{ANT}}$ by combining the greenhouse gas and aerosol fingerprints from each model. 
To ensure comparability with the observations $Y$, the fingerprints were masked using the same missing 
data pattern, and the same spatial aggregation and temporal averaging procedures 
were applied to reduce dimensionality.

Control simulations were drawn from 29 preindustrial control runs contributed by CMIP6 global 
climate models, with lengths ranging from about 100 to 1200 years. To account for model drift, 
a long-term linear trend was removed at each grid box within each control run. To increase 
the effective sample size for covariance estimation, we assumed temporal stationarity and 
partitioned each control run into nonoverlapping 70-year blocks corresponding to 
the 1951–-2020 analysis window. This yielded $m=151$ independent replicates 
for estimating the covariance matrix $\Sigma$. Each block was processed in the 
same way as $Y$, including missing data masking and dimension reduction.

\begin{figure}[tbp]
  \centering \includegraphics[width=5.6in]{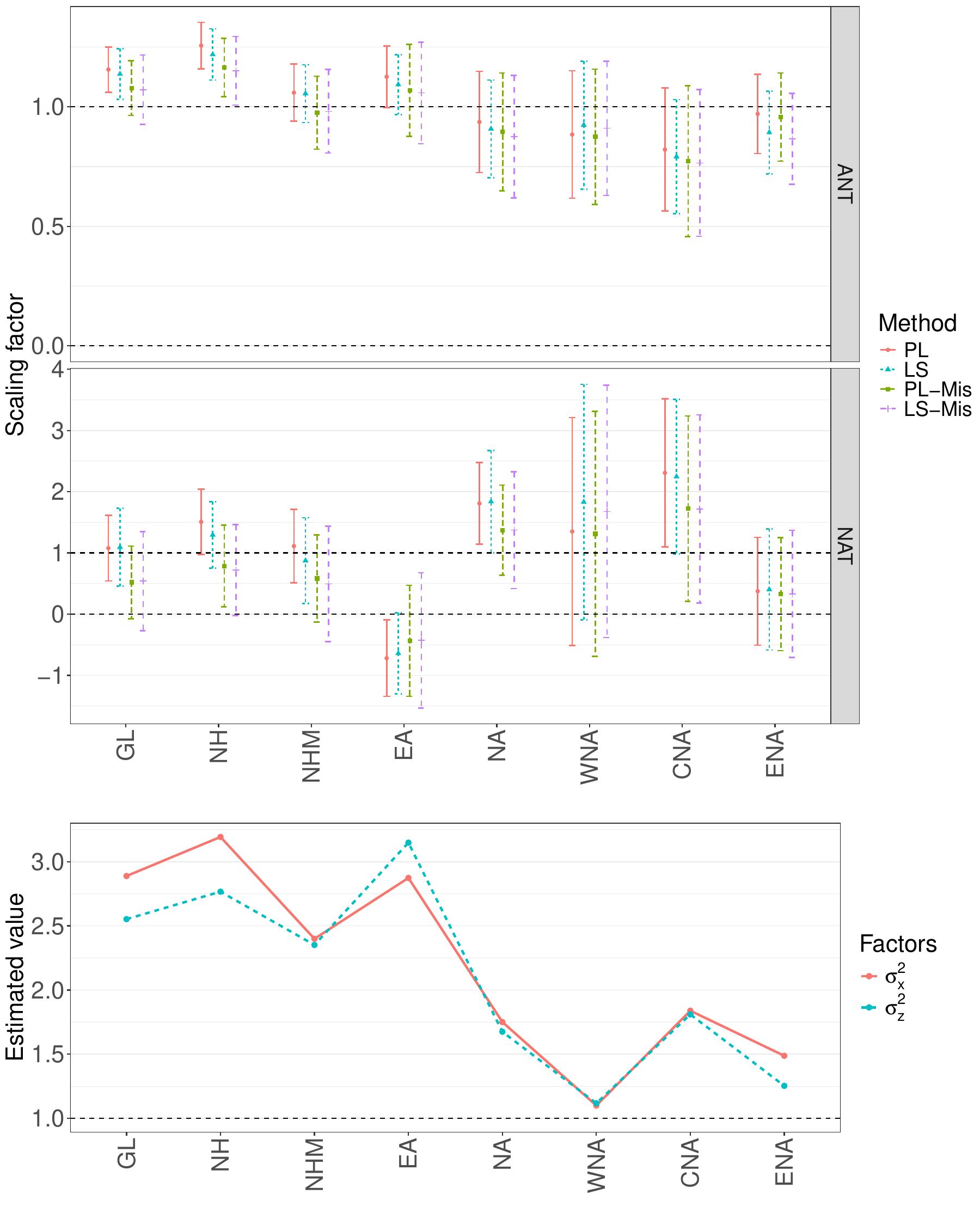}
  \caption{Estimated scaling factors for ANT and NAT forcings that best 
    match the observed 1951--2020 annual mean temperature across different 
    spatial domains, with corresponding 90\% confidence intervals. 
    For weight matrix construction, ``PL'' denotes the proposed framework and 
    ``LS'' the optimal linear shrinkage estimator, both using estimated 
    inflation factors $\sigma_x$ and $\sigma_z$; ``PL-Mis'' and ``LS-Mis'' 
    refer to the cases where we set $\sigma_x^2 = \sigma_z^2 = 1$. 
    The estimated inflation factors for each region are also reported.}
  \label{fig:real_data}
\end{figure}

Figure~\ref{fig:real_data} presents the estimated scaling factors for the ANT and NAT forcings, 
together with their 90\% confidence intervals. Results are shown for the 
proposed framework (PL) and the linear shrinkage estimator (LS), both using estimated 
inflation factors, as well as for the misspecified cases (PL-Mis and LS-Mis) where the 
inflation factors were fixed at unity. The estimated inflation factors $\sigma_x$ and $\sigma_z$, which 
quantify the variation differences between climate model simulations and the observed data, are 
also reported for each region, as they play a central role in the resulting inference.

We first note that the estimated inflation factors are generally not close to the conventional 
value of one often assumed in climate detection and attribution analyses. In most regions, 
climate model simulations tend to overestimate internal variability relative to the observations, 
with two exceptions: WNA, where both factors are close to one, and ENA, where the factors are moderately inflated. 
In these cases, results obtained with estimated factors and those assuming unity are relatively similar. 
When the estimated inflation factors are substantially larger than one, however, the detection 
results change considerably. For instance, at the continental scale GL and NH, the ANT scaling 
factor has a confidence interval entirely above one, indicating that model-simulated fingerprints 
underestimate the observed anthropogenic response, while the NAT scaling factor has an interval 
above zero, implying detectable natural contributions.

As expected from the theoretical results, misspecification of the variance inflation factors leads to biased estimators of the scaling factors and inaccurate assessment of the corresponding asymptotic variance. In this study, confidence intervals based on the estimated inflation factors are consistently shorter than those from the misspecified methods. Properly estimating the inflation factors improves both the accuracy and the precision of the detection results. 

As for the comparison between the proposed PL framework and the LS method, the point estimates are generally 
very close across all analyses. The main difference lies in interval length: the PL method consistently 
yields shorter confidence intervals, reflecting its improved efficiency. In terms of detection and 
attribution conclusions, results are largely consistent across the methods for both ANT and NAT forcings. 
One exception is in regions such as NA and CNA, where the PL estimates for NAT lie above one, suggesting that the 
model-simulated fingerprints somewhat underestimate the observed natural response. Although it is not possible to definitively determine which method is more accurate for this single analysis, given the unknown true underlying scaling factors, 
we emphasize that the comparative performance of these methods has been thoroughly evaluated
in the simulation studies.

Finally, we evaluated the adequacy of the assumed linear model using the residual consistency test introduced in Section~\ref{sec:methodology}. The test examines whether the linear relationship between fingerprints and the observed response, as well as the assumed heterogeneous variance structure, are appropriately specified. Across the subcontinental regions (WNA, CNA, and ENA), the test did not reject the null hypothesis at the 0.05 significance level, suggesting that the linear fingerprinting framework adequately captures the relationship between signals and observations at regional scales. In contrast, for the larger continental domains (GL, NH, NHM, and EA), the test statistics showed mild departures from the null at the 0.05 level. These departures likely reflect the heterogeneity in inflation factors across regions and time periods—consistent with our simulation findings—since larger domains aggregate signals with differing levels of internal variability. For example, while WNA exhibits inflation factors close to one, CNA shows substantially higher values (around 1.75), indicating that the assumption of constant inflation factors may not hold uniformly across space. Overall, the proposed framework provides a reasonable approximation for regional analyses, but future work should explore extensions that allow spatially varying inflation factors and more robust residual diagnostics that remain stable under limited ensemble sizes.

\section{Discussion}
\label{sec:disc}

Fingerprinting is the principal methodological framework for detection and attribution in climate change research and has had a profound influence on modern climate analysis. It underpins observationally constrained climate projections and supports the estimation of key climate system parameters, such as climate sensitivity. The classical framework estimates scaling factors by minimizing total uncertainty in observed responses and noisy fingerprints of external forcings, achieving optimality only under the idealized assumption that the covariance matrix is known a priori. In practice, however, the covariance must be estimated, and this added uncertainty undermines performance. Moreover, variance heterogeneity between noisy fingerprints and responses is frequently observed, introducing systematic bias that existing approaches do not adequately address.

Our proposed methodology overcomes these limitations and advances fingerprinting in three key aspects. First, we introduce a flexible yet easy-to-implement nonlinear shrinkage framework for regularizing the spectrum of the estimated covariance matrix. The shrinkage functions are designed to enable a data-driven, computationally efficient procedure for consistently estimating the asymptotic covariance of the scaling factor estimators. The optimal shrinkage is selected to minimize total variability, using only the estimated fingerprints and the sample covariance matrix from control runs. This facilitates the construction of confidence intervals for scaling factors with valid asymptotic coverage. Second, to account for variance heterogeneity, we augment the fingerprinting framework with variance inflation factors. We propose simultaneous estimation of these factors, establish their consistency, and demonstrate that the resulting scaling factor estimates remain asymptotically normal with properly calibrated variance. Third, we develop a systematic framework for uncertainty quantification and propose a residual diagnostic procedure for assessing model adequacy. This latter component addresses a critical but underexplored issue in fingerprinting, and our contribution is to rigorously establish the consistency of the proposed residual test.

Compared to existing methods in the literature, the proposed approach offers a comprehensive and practical detection and attribution tool, producing point estimates with reduced mean squared error (MSE) and confidence intervals with improved coverage. Unlike approaches that rely on restrictive distributional assumptions \citep{hannart2014optimal,katzfuss2017bayesian} or impose temporal stationarity to mitigate undercoverage at the cost of efficiency, our method introduces no additional assumptions and fully exploits the spatial–temporal covariance structure of $\Sigma$. This ensures more reliable and robust inference. In addition, the method is resilient to variance inflation and computationally efficient relative to Bayesian and bootstrap-based alternatives. Simulation studies in Section~\ref{sec:numerical} show that our approach consistently achieves lower MSE and substantially shorter confidence intervals, highlighting its superiority over existing methods.

In our application to annual mean temperature over multiple continental and subcontinental regions, we find that the estimated inflation factors differ markedly from the conventional unity assumption, indicating that climate models often overestimate internal variability relative to observations. Accounting for this discrepancy leads to notably different attribution results in several regions and consistently shorter confidence intervals, consistent with the theoretical predictions. We also find that detection and attribution conclusions from the proposed framework and the linear shrinkage method are broadly similar, though the proposed framework achieves improved efficiency and robustness, particularly when the inflation factors deviate substantially from one. Together, these findings underscore the efficiency and robustness of our adaptive framework, achieved by addressing variance heterogeneity and employing optimal nonlinear shrinkage for covariance estimation.

The proposed methodology can be extended in several directions. First, the current inference relies on the normality of observations; an important avenue for future work is to investigate robustness under heavy-tailed measurement errors. Second, the framework assumes spatially and temporally constant scaling factors, whereas in some applications it may be more realistic to allow systematic spatial variation. Extending the methodology to accommodate such variation would be valuable. Third, empirical evidence suggests that in certain applications, the sample covariance is dominated by a few leading spikes. Incorporating a factor model for the measurement error, in which correlations are driven by a small number of latent factors, may improve efficiency. Finally, high-dimensional errors-in-variables regression with limited auxiliary data arises in diverse domains such as genomics, econometrics, and neuroimaging. In these settings, where covariates are contaminated by structured noise and precise inference is essential, our regularized framework provides a principled foundation for enhancing estimator efficiency and uncertainty quantification.

\bibliographystyle{apalike}
\bibliography{cited}

\appendix
\section{Extension to General Shrinkage Functions} \label{sec:extension_general_shrinkage}
While polynomial shrinkage functions are well suited for general applications, situations may arise—particularly when additional information about the covariance matrix $\Sigma$ is available—where more specialized forms are advantageous. The proposed methodology extends naturally to such specialized shrinkage functions, preserving the key theoretical guarantees while enabling more targeted forms of regularization.

We say that a general function \(f: \mathbb{R} \to \mathbb{R}\) is \emph{admissible} if it satisfies the following conditions:
\begin{itemize}
    \item[(G1)] \(f\) is non-decreasing on \([0, \psi_1]\) and satisfies \(1 < f(0) \leq f(\psi_1) \leq \psi_2\);
    \item[(G2)] \(f\) is analytic on \([0, \psi_1]\).
\end{itemize}
This definition naturally extends the notion of admissible polynomials introduced in Definition~\ref{def:admissible}.

For an admissible function $f$, define the regularized sample covariance matrix $f(S)$ as 
\[ f(S) = Q \operatorname{Diag}(f(\lambda_1), f(\lambda_2), \cdots f(\lambda_N))Q^T,\]
where $\lambda_1\geq \lambda_2 \geq \cdots \geq \lambda_N$ are the ordered eigenvalues of $S$ and $Q$ is the associated eigenvector matrix.

Suppose a closed set $\mathcal{G}$ of admissible functions is given. Due to the analytic continuation theorem in complex analysis, there exists a contour \( \mathcal{C} \subset \mathbb{C} \) that encloses \( [0, \psi_1] \), such that any \( f \in \calG \) admits an analytic extension to the interior of \( \mathcal{C} \). Let $\calC$ be any such contour. Recall the definition of auxiliary functions $\Delta$, $\Pi$, $\Omega$, and $K$ for polynomial shrinkage in \eqref{eq:def_Delta}--\eqref{eq:def_K}. We extend the definition to a general admissible function $f$ as 
\begin{align*}
&\Delta(f) = \frac{-1}{2\pi i} \oint_{\calC} f^{-1}(z)\delta(z) dz,\\
&\Pi(f) = \frac{-1}{2 \sigma^2_Z \pi i} \oint_{\calC} f^{-1}(z) \pi(z) dz \\
&\Omega(f) = \frac{1}{\sigma^2_Z(2\pi i)^2}  \oint_{\calC} \oint_{\calC} f^{-1}(z_1) f^{-1}(z_2) \omega(z_1,z_2)dz_1 dz_2,\\
&K(f) =  \frac{1}{\sigma^4_Z(2\pi i)^2}  \oint_{\calC}  \oint_{\calC}f^{-1}(z_1) f^{-1}(z_2) \kappa(z_1,z_2)dz_1 dz_2.     
\end{align*}
Here, the integral is taken over the contour \(\mathcal{C}\) in the counterclockwise direction, and \(i\) denotes the imaginary unit. Notably, by \emph{Cauchy's residue theorem}, when \(f\) is a polynomial function, this definition matches with the original formulation.

The proposed methodology and the theoretical results presented in Section~\ref{sec:methodology} remain valid when the weight matrix is taken as $W = f(S)$ where the shrinkage function \( f \) is drawn from the general family \( \mathcal{G} \), and the auxiliary functions \( \Delta \), \( \Pi \), \( \Omega \), and \( K \) are extended accordingly.

Compared to the polynomial shrinkage, the estimation procedure for a general admissible function \( f \) is more complex and computationally less efficient, as it requires numerical evaluation of contour integrals. Overall, we recommend using the polynomial family in applications, unless domain-specific knowledge motivates a particular alternative, as it offers a favorable balance between flexibility and computational efficiency.

\section{Implementation Details}\label{appendix:sec:implement_details}

This appendix provides practical guidance on setting the hyperparameters $\psi_1$ $(\xi)$, $\psi_2$, and $\psi_3$, which constrain the admissible family of polynomial shrinkage functions introduced in Definition~\ref{def:admissible}. 
These constants do not directly tune the estimator but instead determine the feasible region for the polynomial coefficients $(\ell_k,\dots,\ell_1)$, thereby ensuring numerical stability, boundedness, and monotonicity of the shrinkage function $f$ and enforcing proper shrinkage of the empirical covariance spectrum.

\paragraph{Choice of $\psi_1$ ($\xi$).}
As discussed in Section~\ref{subsec:selection_shrinkage}, $\psi_1$ specifies the upper bound of the eigenvalue domain over which the shrinkage function $f$ is defined. 
In practice, we set
\[
\psi_1 = \xi\, \lambda_{\max}(S),
\]
where $\xi>1$ is a small inflation factor that slightly extends the spectral domain beyond the empirical support. 
Values of $\xi$ in the range $[1.05,1.1]$ provide stable results across a wide range of covariance structures and ensemble sizes.

\paragraph{Choice of $\psi_2$.}
The constant $\psi_2$ bounds the maximum value of $f$ at $\psi_1$, constraining the polynomial coefficients to produce a shrinkage covariance estimator that remains well-conditioned and numerically stable. 
We recommend setting
\[
\psi_2 = c_\psi\, \psi_1,
\]
where $c_\psi\in[0.5,1]$ enforces shrinkage of the largest eigenvalues. 
Empirically, $\psi_2 = 0.7\,\psi_1$ achieves a favorable balance between bias reduction and variance control in both moderate- and high-dimensional regimes.

\paragraph{Choice of $\psi_3$.}
The constant $\psi_3$ imposes a minimum separation among the polynomial roots of $f$, ensuring numerical stability in the partial-fraction expansion used in Section~\ref{subsec:selection_shrinkage}. 
A small value such as $\psi_3=10^{-3}$ is sufficient in double-precision arithmetic and has negligible impact on the resulting estimator. 
We recommend keeping $\psi_3$ fixed at this level unless extremely high-degree polynomials ($k>5$) are employed.

\bigskip

\begin{lemma}\label{lemma:explicit_formula_of_beta}
The TLS estimator $\hat{\beta} =  \hat{\beta}(\sigma^2, W^{-1})$ defined in \eqref{eq:def_TLS} admits a close-form expression as
\[\hat{\beta} = \Big[\frac{1}{N}\tilde{X}^T W^{-1} \tilde{X} - \lambda (\sigma^2) D\Big]^{-1} \frac{1}{N} \tilde{X}^T W^{-1} Y,\]
where $\lambda(\sigma^2)$ is the smallest eigenvalue of the matrix 
\begin{equation}
    \label{eq:def_A_mat}
    A(\sigma^2) =  \frac{1}{N} \left[\begin{matrix} I_p & &0\\ 0& & \sigma \end{matrix} \right]\left[ \begin{matrix} D^{-1/2}\tilde{X}^T \\ Y^T  \end{matrix} \right] [f(S)]^{-1} \left[\begin{matrix} \tilde{X}D^{-1/2}, &\, Y  \end{matrix} \right] \left[\begin{matrix} I_p && 0\\ 0 && \sigma \end{matrix} \right].
\end{equation}
\end{lemma}
Note that 
\[
A_{0} = D^{-1/2} \tilde{X}^\top [f(S)]^{-1} \tilde{X} D^{-1/2}
\] 
is the \(p \times p\) leading principal submatrix of \(A(\sigma^2)\). Let \(\lambda_0\) denote the smallest eigenvalue of \(A_0\). It is straightforward to verify that, with probability one, \(\lambda(\sigma^2)\) is a strictly increasing function of \(\sigma^2\) and satisfies \(\lambda(\sigma^2) - \lambda_0 \to 0\) as \(\sigma^2 \to \infty\).  We recommend to choose the upper bound \(\Upsilon\) in \eqref{eq:def_widehat_sigma} such that 
\[
\lambda(\Upsilon) = 0.95 \lambda_0.
\]

In the residual consistency test described in Theorem~\ref{thm:adequacy_test}, the sum of weighted residuals satisfies
\[
\epsilon^\top W^{-1}\epsilon = \lambda(\hat{\sigma}^2_X)\{1 + \hat{\sigma}^2_{X}\hat{\beta}^\top D \hat{\beta}\}.
\]

\begin{lemma}\label{lemma:ledoit_wolf_w}
The linear shrinkage estimator of $\Sigma$ by \citet{ledoit2004well} takes the form
\[ \hat{\Sigma} = w_1 S + w_0 I_N, \quad \mbox{where}\]
\[w_0 = q N^{-1}\tr(S), \quad  w_1 =1-q, \quad  \mbox{and}\quad q = \min\left(1, \frac{1}{m^2\tr[(S - u I_N)^2] } \sum_{j=1}^m \tr[(Z_jZ_j^T -S )^2]\right).\]
The initial weight matrix in Algorithm \ref{algo:estimation} is taken as $\hat{\Sigma}  = w S + I_N$, where $w = w_1/w_0$.     
\end{lemma}

\newpage\clearpage


\setcounter{page}{1}
\setcounter{section}{0}
\renewcommand{\thesection}{S.\arabic{section}}
\setcounter{subsection}{0}
\renewcommand{\thesubsection}{S.\arabic{section}.\arabic{subsection}}
\setcounter{equation}{0}
\renewcommand{\theequation}{S.\arabic{equation}}
\setcounter{figure}{0}
\setcounter{table}{0}
\begin{center}
\textbf{ \large Supplementary Material to \\ ``Adaptable Fingerprinting with Nonlinear Shrinkage for Climate Change Detection and Attribution under Variance Heterogeneity''}\\
{Haoran Li} and {Yan Li}\\
{Department of Mathematics and Statistics}\\
{Auburn University}
\end{center}

In this supplementary material, we provide detailed proofs of the main technical results in the manuscript.  Since polynomial shrinkage is a special case of the general admissible analytic shrinkage introduced in Section~\ref{sec:extension_general_shrinkage} of the Appendix, all proofs are developed directly under this general shrinkage framework. Throughout, we assume that $f$ is an admissible shrinkage function satisfying conditions (G1) and (G2). For notational convenience, we denote $g(x) = f^{-1}(x)$.

Select $\calC$ to be any contour enclosing the interval $[0, \psi_1]$, such that the shrinkage function $f$ admits an analytic extension to the interior of $\calC$. For concreteness, we take $\calC$ to be the rectangle with vertices $\underline{u}_0 \pm i v_0$ and $\overline{u}_0 \pm i v_0$, where $v_0 > 0$, $\overline{u}_0 > \alpha$, and $\underline{u}_0 < 0$. Such a rectangle always exists. 
For $z \in \calC$, the resolvent $(S - z I_N)^{-1}$ is well defined, except in the trivial case when an eigenvalue of $S$ coincides with $\overline{u}_0$. Since this occurs with probability zero, we exclude the two points $\overline{u}_0 + 0i$ and  $\underline{u}_0 + 0i$ from the contour. This exclusion does not affect the evaluation of contour integrals of the relevant processes over $\calC$. 
Finally, let $\psi_0$ be a constant strictly between $\psi_1$ and $\limsup_{N \to \infty} \sigma_Z^2 \lambda_1(\Sigma) (1+\sqrt{\gamma})^2$.

\section{Technical Lemmas}
\label{sec:technical_lemmas}

This section collects several technical lemmas that will be used in the subsequent proofs. Throughout, $A^T$ denotes the transpose of a real matrix $A$, or the conjugate transpose when $A$ is complex. For a complex number $z$, we write $\Re(z)$ and $\Im(z)$ for its real and imaginary parts, respectively. The notation $\|\cdot\|_F$ denotes the Frobenius norm of a matrix.  

We use the probabilistic order notation $O_p(1)$ and $o_p(1)$ in the following sense: $O_p(1)$ may refer to a scalar random variable bounded in probability, or to a random matrix $A$ such that $\|A\|_2 = O_p(1)$. Similarly, $o_p(1)$ may denote a scalar random variable converging to zero in probability, or a random matrix $A$ with $\|A\|_2 = o_p(1)$. The notation $\calK_k$ denotes a universal constant that only depends on the parameter $k$ and may differ from line to line.

\begin{lemma}[Woodbury Matrix Identity]\label{lemma:woodbury}
The following identity holds:
\[
(A + UCV)^{-1} = A^{-1} - A^{-1} U \left( C^{-1} + V A^{-1} U \right)^{-1} V A^{-1},
\]
for matrices \( A, U, C, V \) of conformable sizes, assuming all inverses exist and are well-defined.
\end{lemma}

\begin{lemma}[Burkholder's inequality]\label{lemma:burkholder}
    Let $\{Y_i\}$ be a complex martingale difference sequence with respect to the increasing $\sigma$-field $\{\sigma_i\}$. Then, for $k\geq 2$,
    \[ \mE |\sum_{i} Y_i|^k \leq \calK_k \mE \Big(\sum_{i} \mE ( |Y_i|^2 \mid \sigma_{i-1})  \Big)^{k/2}  + \calK_k \mE (\sum_{i} |Y_i|^k ), \]
    where $\calK_k$ is a constant depending only on $k$. 
\end{lemma}

\begin{lemma}\label{lemma:probabiltiy_bound_concentration_quadratic_form}
    Suppose that $W \sim \calN(0, I_N)$, and let $A$ be a Hermitian matrix with $\|A\|_2 < \calK$. Then, for all $0<t <\calK$,
    \[ \mP \left( \frac{1}{N} |W^T A W -\tr(A)|  >t \right) \leq 2 \exp\Big( -\frac{Nt^2}{4\calK^2}\Big). \]
\end{lemma}

\begin{lemma}[Lemma 2.7 of \cite{bai1998no}]\label{lemma:concentration_quadratic_forms}
Let $W = (w_1, \dots, w_N)^T$, where $w_i$'s  are i.i.d. real r.v.'s with mean $0$ and variance $1$. Let $B$ be a deterministic (real or complex) matrix. Then, for any $j\geq 2$, we have 
\[\mE |W^T BW - \tr (B)|^j \leq \calK_j(\mE w_1^4 \tr(BB^T))^{j/2} + \calK_j \mE w_1^{2j} \tr[ (BB^T)^{j/2}],\]
where $\calK_j$ is a constant only depending on $j$. 
\end{lemma}

\begin{lemma}
    \label{lemma:Mcdiamid}
    Suppose that $f: \mathbb{C}^N\to\mathbb{C}$ satisfies the bounded differences in the sense that 
    \[ \sup_{x_i' \in \mathbb{C} } \Big| f(x_1, x_2, \cdots, x_{i-1}, x_i, x_{i+1}, \cdots, x_N)   - f(x_1, x_2, \cdots, x_{i-1}, x'_{i}, x_{i+1}, \cdots, x_N )\Big|\leq c,\]
    for any $i$ and any $x_1, \dots, x_{i-1}, x_{i}, x_{i+1}, \dots, x_N$.
    Consider independent random variables $X_1, X_2, \dots X_N$. Then, for any $\varepsilon>0$,
    \[ \mP\left( \left|f(X_1, X_2, \cdots, X_N) - \mE f(X_1, X_2, \cdots, X_N )\right| \geq \varepsilon \right) \leq 2 \exp\Big(- \frac{2\varepsilon^2}{c^2N} \Big). \]
\end{lemma}

Recall that $S$ is the sample covariance matrix of the control runs $Z_j$'s.  We now introduce a leave-one-out version of \( S \) in which the dependence on a specific \( Z_j \) is removed. For each \( j = 1, 2, \dots, m \), define
\[
S_j = \frac{1}{m} \sum_{i \neq j} Z_i Z_i^T.
\]

The following lemma provides a deterministic bound on the difference between the resolvent of $S$ and $S_j$ when $z$ is away from the real line. 
\begin{lemma}[Lemma 2.10 of \citet{bai1998no}] \label{lemma:residual_Sigma_k_Sigma}
Let \( z \in \mathbb{C} \) with $\Im(z) \neq 0$. For any Hermitian matrix \( D \), we have
\[
\left| \operatorname{tr}\left[(S - z I_N)^{-1} D\right] - \operatorname{tr}\left[(S_j - z I_N)^{-1} D\right] \right| 
\leq \frac{ \|D\|_2 }{ |\Im(z)|}.
\]
In case it is known that the largest eigenvalue of $S$ is contained in $[0, \psi_0]$, we have for $z\in\calC$, 
\[
\left| \operatorname{tr}\left[(S - z I_N)^{-1} D\right] - \operatorname{tr}\left[(S_j - z I_N)^{-1} D\right] \right| 
\leq \|D\|_2 \|(S_j - z I_N)^{-1}\|\leq \calK \|D\|_2.
\]
\end{lemma}
\noindent Here, the second part is proved using similar arguments as those in \citet{bai1998no}. Indeed, 
\begin{align*}
\left| \operatorname{tr}\left[(S - z I_N)^{-1} D\right] - \operatorname{tr}\left[(S_j - z I_N)^{-1} D\right] \right| & = \left|\frac{ m^{-1} Z_j^T (S_j - zI_N)^{-1} D(S_j - zI_N)^{-1} Z_j}{1+ m^{-1} Z_j^T (S_j - zI_N)^{-1}Z_j  }\right|\\
&\leq \Big\| (S_j - zI_N)^{-1/2} D (S_j - zI_N)^{-1} \Big\|_2\leq \|D\| \|(S_j - z I_N)^{-1} \|_2.
\end{align*}
Here, notice that the largest eigenvalue of $S_j$ is no larger than that of $S$. 

Next, we present the renowned Mar\v{c}enko-Pastur equation in the Random Matrix Theory (RMT) literature. 
\begin{lemma}
    \label{lemma:M_P_theorem}
Suppose that $L^{\Sigma}$ is a compactly supported distribution, nondegenerate to zero. For any $z\in\mathbb{C}^+ \coloneqq \{u + iv\in\mathbb{C}: v>0\}$, there exists a unique solution $s(z) \in \mathbb{C}^+$ to the equation
\[ s(z) = \int \frac{dL^{\Sigma}(\tau)}{\sigma^2_Z\tau \{1 -\gamma - \gamma z s(z)\} - z}, \]
commonly referred to as the Mar\v{c}enko-Pastur equation.  We further extend the domain of $s(z)$ from $\mathbb{C}^+$ to $\mathbb{C}^-\coloneqq \{u+iv\in\mathbb{C}: v<0 \}$ by setting $s(z) = \overline{s(\overline{z})}$ if $z\in\mathbb{C}^-$. 

The following analytic properties of $s(z)$ are known.  First, the function $s(z)$ is analytic in $\mathbb{C}\setminus\mathbb{R}$. Secondly, $s(z)$ is well-behaved when $z$ is on $\calC$ in the following sense: 
\begin{equation}\label{eq:lower_bound1}
\begin{aligned}
\inf_{z\in\calC} ~&| 1- \gamma -\gamma zs(z)| > \calK,\\
\inf_{z\in\calC, \tau \in [0,\psi_0]}& |\sigma^2_Z \tau \{1-\gamma - \gamma z s(z) \}  -z |>\calK,
\end{aligned}
\end{equation}
where $\calK>0$ is a general constant. The results indicate that the integrand on the right-hand of the Mar\v{c}enko-Pastur equation is well-behaved when $z\in \calC$.  
\end{lemma}
\noindent For details, see \citet{bai2010spectral} and a discussion in Section S.2 of \citet{Li2020high}.

\section{Pointwise Convergences of Processes}\label{sec:pointwise_convergence}

Consider the resolvent $(S - z I_N)^{-1}$, $z\in\calC$. Let 
\[\theta_\infty (z) = 1/(1- \gamma - \gamma z s(z)), \quad z\in\calC.\] 
Note that $\theta_\infty(z)$ is well-defined due to the analytic properties of $s(z)$ presented in Lemma \ref{lemma:M_P_theorem}.
It is well-known in RMT that the resolvent has a \emph{deterministic equivalent}
\[ J(z) = \Big(\frac{\sigma_Z^2}{\theta_\infty(z)}\Sigma - zI_N\Big)^{-1}  \]
in the following sense. 
\begin{lemma}
    \label{lemma:deterministic_equivalent}
Suppose that \ref{enum:HD_regime}--\ref{enum:converge_X_Sigma_1} hold. Consider any fixed $z\in\mathbb{C}\setminus\mathbb{R}$. For any sequence of matrices $A$ such that $\|A\|_2 <\infty$, we have 
\[ \frac{1}{N}\tr[(S- zI_N)^{-1}A] - \frac{1}{N}\tr[J(z) A] = o_p(N^{-1/2}).\]
Secondly, for any sequence of vectors $a$ such that $N^{-1} \|a\|_2^2 <\infty$,
\[\frac{1}{N} a^T(S - zI_N)^{-1}a - \frac{1}{N}a^TJ(z)a  = o_p(N^{-1/2}). \]
\end{lemma}
\noindent See, for example, \citet{knowles2017anisotropic} and \citet{karoui2011geometric} for the proof of Lemma \ref{lemma:deterministic_equivalent}. 

Define 
\begin{align*}
&\pi_\infty = (1/\gamma) (\theta_\infty(z)-1),\\
&\delta_\infty(z) =\lim_{N\to\infty} \frac{1}{N} X^T \Big(\frac{\sigma^2_Z}{\theta_\infty(z)} \Sigma - z I \Big)^{-1}X.
\end{align*}
Note that the existence of $\delta_\infty(z)$ follows from Condition \ref{enum:converge_X_Sigma_1}. 

Notice that the Mar\v{c}enko-Pastur equation implies that
\[ \frac{1}{N}\tr[J(z)] = \int\frac{dF^{\Sigma}(\tau)}{ \sigma_Z^2\tau \{1 -\gamma -\gamma zs(z) \} -z}  =  \int\frac{dL^{\Sigma}(\tau)}{ \sigma_Z^2\tau \{1 -\gamma -\gamma zs(z) \} -z} +o(N^{-1/2}) = s(z) + o(N^{-1/2}),\]
\[ \frac{\sigma_Z^2}{N} \tr[J(z)\Sigma] = \int \frac{\sigma_Z^2 \tau dF^{\Sigma}(\tau)}{ \sigma_Z^2\tau \{ 1 - \gamma - \gamma z s(z)\} -z} =  \int \frac{\sigma_Z^2 \tau dL^{\Sigma}(\tau)}{ \sigma_Z^2\tau \{ 1 - \gamma - \gamma z s(z)\} -z} + o(N^{-1/2}) = \pi_\infty(z) + o(N^{-1/2}).\]
The following results are direct consequences of Lemma \ref{lemma:deterministic_equivalent}. 
\begin{corollary}\label{corollary:converge_first_order}
Under the settings of Lemma \ref{lemma:deterministic_equivalent}, setting $A= I_p$, we have 
\[ \mu(z) - s(z) = o_p(N^{-1/2}),\]
\[\pi(z) - \pi_\infty(z) = o_p(N^{-1/2}).\]
Setting $A = \sigma_Z^2\Sigma$, we have 
\[ \frac{\sigma_Z^2}{N} \tr[(S-zI_N)^{-1}\Sigma] -\pi_\infty(z) = o_p(N^{-1/2}). \]
Setting $a= X_i \pm X_j$ for any pair $i,j$, we obtain
\[ \frac{1}{N} X^T(S - zI_N)^{-1} X  - \delta_\infty(z) = o_p(N^{-1/2}). \]
\end{corollary}
Next, we consider the bias induced by replacing $X$ with $\tilde{X}$ in the quadratic form $N^{-1} X^T (S - z I_N)^{-1}X$. Note that $\tilde{X}$ can be expressed as $\tilde{X} = X + \sigma_X \Sigma^{1/2}W D^{1/2}$, where $W$ is a matrix of iid $\calN(0,1)$ entries. We obtain 
\begin{align*}
\frac{1}{N}\tilde{X}^T (S - z I_N)^{-1}\tilde{X} = &\frac{1}{N} X^T (S - zI_N)^{-1} X + \frac{1}{N} \sigma_X D^{1/2}W^T \Sigma^{1/2} (S- zI_N)^{-1} X \\&+ \frac{1}{N} \sigma_X X^T (S- zI_N)^{-1} 
\Sigma^{1/2} W D^{1/2} 
 + \frac{1}{N} \sigma_X^2 D^{1/2} W^T \Sigma^{1/2} (S-zI_N)^{-1} \Sigma^{1/2}WD^{1/2}. 
\end{align*}
Using Lemma \ref{lemma:concentration_quadratic_forms}, we have for any $k\geq 2$
\[\mE\Big|\frac{1}{\sqrt{N}} W_i^T \Sigma^{1/2} (S - z I_N)^{-1} \Sigma^{1/2} W_i -  \frac{1}{\sqrt{N}} \tr [(S- zI_N)^{-1} \Sigma] \Big|^k \leq \calK_k \mE \| (S - zI_N)^{-1} \|^k \leq  \frac{\calK_k}{|\Im(z)|^k},\]
where $W_i$ is the $i$th column of $W$. Moreover, if $i\neq j$,
\[\mE\Big|\frac{1}{\sqrt{N}} W_i^T \Sigma^{1/2} (S - z I_N)^{-1} \Sigma^{1/2} W_j \Big|^k \leq \calK_k \mE \| (S - zI_N)^{-1} \|^k \leq \frac{\calK_k}{|\Im(z)|^k}.\]
Also, for any $i,j$, 
\[ \mE \Big| \frac{1}{\sqrt{N}} W_i^T \Sigma^{1/2} (S - z I_N)^{-1} X_j \Big|^{k}  \leq {\calK_k(N^{-1} X_j^TX_j)^{k/2}} \mE \|(S- zI_N)^{-1} \|^{k} \leq  \frac{\calK_k \|N^{-1}X^TX\|^{k/2}_2 }{|\Im(z)|^k }. \]
Recall that 
\[ \delta(z) = \frac{1}{N}\tilde{X}^T (S- zI_N)^{-1} \tilde{X}  - (\sigma^2_X /\sigma^2_Z) \pi(z)D.\]
Applying Lemma \ref{lemma:concentration_quadratic_forms}, we obtain the following lemma. 
\begin{lemma}\label{lemma:converge_Xtilde}
Suppose that \ref{enum:HD_regime}--\ref{enum:converge_X_Sigma_1} hold. Consider any fixed $z\in\mathbb{C}\setminus\mathbb{R}$. We have
\[ \frac{1}{N}\tilde{X}^T (S- zI_N)^{-1} \tilde{X} = \frac{1}{N} X^T (S-zI_N)^{-1}X + \frac{\sigma^2_X}{N} \tr[(S-zI_N)^{-1}\Sigma] D + O_p(N^{-1/2}).    \]
\begin{equation}\label{eq:converge_delta_delta_infty}
\delta(z) = \frac{1}{N}X^T (S- z I_N)^{-1} X + O_p(N^{-1/2}) = \delta_\infty(z) + O_p(N^{-1/2}).    
\end{equation}
\end{lemma}

The results in Corollary \ref{corollary:converge_first_order} and Lemma \ref{lemma:converge_Xtilde} can be further generalized as follows. 
\begin{lemma}
    \label{lemma:converge_derivatives}
    Under the settings of Lemma \ref{lemma:deterministic_equivalent}, we have 
    
    \[\mu'(z) =  s'(z) + o_p(N^{-1/2}), \quad  \pi'(z) =  \pi'_\infty(z) + o_p(N^{-1/2}),\]
    \[ \frac{1}{N} X^T(S - zI_N)^{-2} X   = \delta'_\infty(z) + o_p(N^{-1/2}),\quad \delta'(z)  =  \delta'_\infty(z) + O_p(N^{-1/2}).\]
\end{lemma}
\noindent The proof of Lemma \ref{lemma:converge_derivatives} can be found in \citet{bai2008clt}  and \citet{karoui2011geometric}.

Moreover, we can show the following results. 
Define 
\[\kappa_\infty(z_1,z_2) = 
\begin{cases}
\theta_\infty(z_1) \theta_\infty(z_2)\cfrac{z_1\pi_\infty(z_1) -z_2 \pi_\infty(z_2)}{z_1 -z_2}, & z_1 \neq z_2,\\
\theta^2_\infty(z_1) \Big[ z_1 \pi'_\infty(z_1) + \pi_\infty(z_1) \Big], &z_1 =z_2. 
\end{cases}
\]
\[\omega_\infty(z_1,z_2) = 
\begin{cases}
\theta_\infty(z_1) \theta_\infty(z_2)\cfrac{z_1\delta_\infty(z_1) -z_2 \delta_\infty(z_2)}{z_1 -z_2}, & z_1 \neq z_2,\\
\theta^2_\infty(z_1) \Big[ z_1 \delta'_\infty(z_1) + \delta_\infty(z_1) \Big], &z_1 =z_2. 
\end{cases}
\]
\begin{lemma}\label{lemma:converge_kappa_omega}
    Suppose that \ref{enum:HD_regime}--\ref{enum:converge_X_Sigma_1} hold. Consider any fixed $z_1$ and $z_2$ with a non-zero imaginary part. We have 
    \[ \frac{\sigma_Z^4}{N}\tr\Big[(S- z_1 I_N)^{-1}\Sigma (S- z_2 I_N)^{-1}\Sigma \Big] = \kappa(z_1, z_2)  + o_p(1) = \kappa_\infty(z_1, z_2)  + o_p(1),\]
     \[ \frac{\sigma_Z^2}{N} X^T (S- z_1 I_N)^{-1}\Sigma (S- z_2 I_N)^{-1}X = \omega(z_1, z_2)  + o_p(1) =  \omega_\infty(z_1, z_2)  + o_p(1).\]
\end{lemma} 
\begin{proof}
When $z_1 = z_2$, the results are shown in \citet{karoui2011geometric} and \citet{li2025regularized}. We only consider the case when $z_ 1 \neq z_2$. Since
\[ z_1(S - z_1 I_N)^{-1} -  z_2 (S - z_2 I_N)^{-1} = (S - z_1 I_N)^{-1}(z_1S - z_2 S)(S - z_2I_N)^{-1}, \]
using Corollary \ref{corollary:converge_first_order} and the Woodbury matrix identity, we have 
\begin{align*}
&\frac{z_1 \pi(z_1) - z_2 \pi(z_2) }{z_1 - z_2} \\
&= \frac{ \sigma_Z^2 N^{-1} z_1\tr[ (S - z_1 I_N)^{-1}\Sigma] - \sigma_Z^2N^{-1}z_2\tr[(S- z_2 I_N)^{-1}\Sigma]}{z_1 -z_2} + o_p(N^{-1/2})\\
& = \frac{\sigma_Z^2}{N}\tr[ (S - z_1 I_N)^{-1} S (S- z_2 I_N)^{-1} \Sigma]  + o_p(N^{-1/2})\\
& =\frac{\sigma_Z^2}{m}\sum_{j=1}^m \frac{1}{N} Z_j^T (S - z_1 I_N)^{-1}\Sigma (S- z_2 I_N)^{-1} Z_j + o_p(N^{-1/2}) \\
& =\frac{\sigma_Z^2}{m}\sum_{j=1}^m \frac{N^{-1} Z_j^T (S_j - z_1 I_N)^{-1}\Sigma (S_j - z_2 I_N)^{-1} Z_j}{\{1+ (1/m) Z_j^T (S_j - z_1 I_N)^{-1}Z_j\} \{1+ (1/m) Z_j^T (S_j - z_2I_N)^{-1}Z_j \} } + o_p(N^{-1/2}).
\end{align*}

Applying Lemma \ref{lemma:concentration_quadratic_forms}, for $k \geq 2$, there exists a constant $\calK_k$ such that 
\begin{align*}
N^{k/2} \mE \Big|\frac{\sigma_Z^2}{N}&Z_j^T (S_j - z_1 I_N)^{-1}\Sigma (S_j - z_2 I_N)^{-1} Z_j -\frac{\sigma_Z^4}{N}\tr[(S_j - z_1 I_N)^{-1}\Sigma (S_j - z_2 I_N)^{-1}\Sigma]\Big|^k \\
&\leq \calK_k \mE \| (S_j - z_1 I_N)^{-1} (S_j - z_2 I_N)^{-1}\|^k_2 \leq \frac{\calK_k}{|\Im(z_1) \Im(z_2)|^k},\\
N^{k/2}\mE \Big| \frac{1}{N} &Z_j^T (S_j - z I_N)^{-1} Z_j - \frac{\sigma_Z^2}{N}\tr[(S_j - z I_N)^{-1}\Sigma]  \Big|^k \leq \calK_k \mE \|(S_j - zI_N)^{-1} \|_2^k \leq \frac{\calK_k}{|\Im z|^k }.  
\end{align*}
Here, we are using the bound $ \|(S - zI_N)^{-1} \| \leq 1/{|\Im(z)|}$.  It follows then
\begin{align*}
   & \frac{z_1 \pi(z_1) - z_2 \pi(z_2) }{z_1 - z_2}\\ &= \frac{\sigma_Z^4}{m}\sum_{j=1}^m \frac{ N^{-1}\tr[(S_j - z_1 I_N)^{-1}\Sigma (S_j - z_2 I_N)^{-1}\Sigma]}{ \{1+ (\sigma_Z^2/m)\tr[(S_j - z_1 I_{N})^{-1}\Sigma]\} \{1+ (\sigma_Z^2/m)\tr[(S_j - z_2 I_{N})^{-1}\Sigma]\}} + O_p(N^{-1/2}).
 \end{align*}

Moreover, using Lemma \ref{lemma:residual_Sigma_k_Sigma}
\[ \Big|\tr[(S_j - z I_N)^{-1} \Sigma] - \tr [ (S- zI_N)^{-1}\Sigma]\Big| \leq \frac{\|\Sigma\|_2}{\Im(z)},\]
\[ \Big|\tr[ (S_j- z_1 I_N)^{-1}\Sigma (S_j -z_2 I_N)^{-1}\Sigma] - \tr[ (S- z_1 I_N)^{-1}\Sigma (S- z_2 I_N)^{-1}\Sigma]  \Big| \leq \frac{\calK\|\Sigma\|_2^2}{|\Im(z_1)\Im(z_2)| }.\]

We conclude that 
\begin{align*}
    \frac{z_1\pi(z_1) - z_2 \pi(z_2)}{z_1 - z_2} = \frac{\sigma_Z^4 N^{-1}\tr[ (S- z_1 I_N)^{-1}\Sigma (S- z_2I_N)^{-1}\Sigma  ]   }{(1 + (N/m) \pi(z_1))(1+ (N/m) \pi(z_2) ) } + O_p(N^{-1/2}).
\end{align*}
Therefore,
\[ \sigma_Z^4 \frac{1}{N} \tr[ (S- z_1 I_N)^{-1}\Sigma (S- z_2 I_N)^{-1}\Sigma] = \kappa(z_1, z_2) + O_p(N^{-1/2}) = \kappa_\infty(z_1, z_2) + O_p(N^{-1/2}).\]
The convergence of $N^{-1} X^T (S - z_1 I_N)^{-1}\Sigma (S - z_2 I_N)^{-1}X$ can be shown following analogous arguments. We omit details. 
\end{proof}


\section{Proof of Lemma \ref{lemma:convergence_Delta_Pi_Omega_K}}\label{sec:convergence_integrals}

Recall the definition of $\Pi(f)$, $\Delta(f)$,  $\Omega(f)$, and $K(f)$ in Section \ref{sec:extension_general_shrinkage} of the Appendix shown here for the reader's convenience
\begin{align*}
    &\Pi(f) = \frac{-1}{2\sigma_Z^2 \pi i} \oint_{\calC} g(z) \pi(z) dz,\\
    &\Delta(f) = \frac{-1}{2\pi i } \oint_{\calC} g(z) \delta(z) dz,\\
    &\Omega(f) = \frac{1}{\sigma_Z^2 (2\pi i)^2} \oint_{\calC}\oint_{\calC} g(z_1) g(z_2) \omega(z_1,z_2)dz_1 dz_2,\\
    &K(f) = \frac{1}{\sigma_Z^4 (2\pi i)^2} \oint_{\calC} \oint_{\calC} g(z_1)g(z_2) \kappa(z_1, z_2)dz_1dz_2.
\end{align*}
Define the limiting version of the objects as 
\begin{align*}
    &\Pi_\infty(f) = \frac{-1}{2\sigma_Z^2 \pi i} \oint_{\calC} g(z) \pi_\infty(z) dz,\\
    &\Delta_\infty(f) = \frac{-1}{2\pi i } \oint_{\calC} g(z) \delta_\infty(z) dz,\\
    &\Omega_\infty(f) = \frac{1}{\sigma_Z^2 (2\pi i)^2} \oint_{\calC}\oint_{\calC} g(z_1) g(z_2) \omega_\infty(z_1,z_2)dz_1 dz_2,\\
    &K_\infty(f) = \frac{1}{\sigma_Z^4 (2\pi i)^2} \oint_{\calC} \oint_{\calC} g(z_1)g(z_2) \kappa_\infty(z_1, z_2)dz_1dz_2.
\end{align*}
It is worth mentioning that the integrands are well-behaved on $\calC$ due to the analytic properties of $s(z)$ on $\calC$ presented in Lemma \ref{lemma:M_P_theorem}.  


First of all, by Cauchy's residue theorem, when all eigenvalues of $S$ are contained in $[0, \psi_0]$, we have
\begin{align*}
&\frac{1}{N} \tr[ g(S) \Sigma] = \frac{-1}{2\pi i} \oint_{\calC} g(z)\frac{1}{N}\tr[ (S- zI_N)^{-1} \Sigma] dz,\\
& \frac{1}{N} X^T g(S) X = \frac{-1}{2\pi i } \oint_{\calC} g(z)\frac{1}{N} X^T (S - zI_N)^{-1} Xdz,\\
& \frac{1}{N}\tr[ g(S)\Sigma g(S)\Sigma] = \frac{1}{(2\pi i)^2}\oint_{\calC}\oint_{\calC} g(z_1) g(z_2) \frac{1}{N}\tr[ (S- z_1 I_N)^{-1}\Sigma (S- z_2 I_N)^{-1}\Sigma] dz_1 dz_2,\\
& \frac{1}{N}X^T g(S)\Sigma g(S)X = \frac{1}{(2\pi i)^2}\oint_{\calC}\oint_{\calC}g(z_1) g(z_2)  \frac{1}{N} X^T(S- z_1 I_N)^{-1}\Sigma (S- z_2 I_N)^{-1}X dz_1 dz_2.
\end{align*}
Due to Lemma \ref{lemma:extreme_eigenvalue_bound}, the event occurs with high probability for all sufficiently large $N$. Therefore, the limits in probability of the traces and quadratic forms are the same as those of the right-hand sides. Therefore, we only need to show the convergence of the integrals of the traces and quadratic forms in the following sense. 
\begin{lemma}\label{lemma:convergence_integrals}
Suppose that \ref{enum:HD_regime}--\ref{enum:converge_X_Sigma_1} hold.  We have 
\begin{align*}
& \oint_{\calC} |\theta(z) - \theta_\infty(z) |  |dz| = o_p(N^{-1/2}), \\  
&\oint_{\calC}  \Big\|\delta(z) - \delta_\infty(z)\Big\|_2 |dz|  = O_p(N^{-1/2}),\\
&\oint_{\calC} \Big| \frac{\sigma_Z^2}{N}\tr[(S- zI_N)^{-1}\Sigma]  - \pi(z) \Big| |dz|  = o_p(N^{-1/2}),\\
&\oint_{\calC}  \Big\| \frac{1}{N} X^T(S - z I_N)^{-1}X  - \delta(z)\Big\|_2 |dz|  = O_p(N^{-1/2}),\\
&\oint_{\calC}\oint_{\calC}\Big|\frac{\sigma_Z^4}{N}\tr[(S- z_1 I_N)^{-1}\Sigma (S - z_2 I_N)^{-1}\Sigma] - \kappa(z_1, z_2) \Big| |dz_1| |dz_2| = o_p(1),\\ 
&\oint_{\calC}\oint_{\calC}\Big\|\frac{\sigma_Z^2}{N}X^T(S- z_1 I_N)^{-1}\Sigma (S - z_2 I_N)^{-1}X - \omega(z_1, z_2) \Big\|_2 |dz_1| |dz_2| = o_p(1).
\end{align*}
\end{lemma}

While the pointwise convergence of the above processes is established in Section \ref{sec:pointwise_convergence}, we only need to show the tightness of the processes on $\calC$. However, it is not immediate. The main reason is that the operator norm of $(S - zI)^{-1}$ may diverge when $z$ approaches the real line and the largest eigenvalue of $S$ exceeds $\psi_0$. Although Lemma~\ref{lemma:extreme_eigenvalue_bound} shows that such events occur with small probability, no explicit probability bound is available, leaving the tightness of the process uncertain. 

To overcome this issue, we demonstrate that tightness can be established after introducing a variable truncation and process smoothing procedure. Crucially, this modification alters the contour integrals only negligibly, and the integral of altered processes has the same limit in probability. 


\emph{Variable truncation}. Recall that $Z_{j}$'s are iid $\calN(0, \sigma_Z^2\Sigma)$. We set $Z_j  = \sigma_Z \Sigma^{1/2} {W}_{j}$, where  ${W}_j$ are iid $\calN(0, I_N)$. Take a random variable ${w}\sim \calN(0,1)$. As shown in \citet{yin1988limit}, there exists a positive sequence $\varepsilon_N$ such that 
\[ \varepsilon_N >0 \quad \mbox{and} \quad \varepsilon_N^{-4} \mE [{w}^4\mathbbm{1}(|{w}| \geq \varepsilon_N N^{1/2})] \longrightarrow 0. \]
We truncate ${w}$ to be ${w} \mathbbm{1}(|{w}| \leq \varepsilon_N N^{1/2})$ and re-standardize the truncated variable to maintain zero mean and unit variance. Namely, we obtain 
\[ \vw  = \frac{{w} \mathbbm{1}(|{w}| \leq \varepsilon_N N^{1/2}) -\mE {w} \mathbbm{1}(|{w}| \leq \varepsilon_N N^{1/2})}{ \{ \mE [ {w} \mathbbm{1}(|{w}| \leq \varepsilon_N N^{1/2}) -\mE {w} \mathbbm{1}(|{w}| \leq \varepsilon_N N^{1/2})]^2 \}^{1/2}  }.\]
It follows then, for some constant $\calK$, when $N$ is sufficiently large, 
\begin{equation}\label{eq:truncated_variable_conditions}
|\vw| \leq \calK \varepsilon_N N^{1/2}, \quad \mE \vw =0, \quad \mE \vw^2 =1, \quad \mE \vw^4 <\infty.
\end{equation}
Apply this truncation and re-standardization step to all entries of $W_j$'s. Denote the modified vectors to be $\vW_j$. We obtained the truncated version of $Z_j$, $S$ and $S_j$'s as 
\[ \vZ_j = \sigma_Z \Sigma^{1/2}\vW_j, \quad \vS = \frac{1}{m} \sum_{j=1}^m \vZ_j \vZ_j^T, \quad \mbox{and} \quad \vS_j = \frac{1}{m} \sum_{i\neq j} \vZ_j \vZ_j^T.\]

The main purpose of truncating the variable is to obtain the following convergence rate on the probability of the largest eigenvalue exceeding $\psi_0$. The results are shown in  \citet{yin1988limit} and \citet{bai2008clt}. 
\begin{lemma}
    \label{lemma:poly_bound_prob_exceed_right}
   Suppose \ref{enum:HD_regime}--\ref{enum:converge_X_Sigma_1} hold. For any positive $l$, 
\[\mP\Big(\lambda_1(\vS) \geq  \psi_0 \Big) = o(N^{-l}).\]
\end{lemma}
\noindent Note that although Lemma \ref{lemma:extreme_eigenvalue_bound} implies $\mP(\lambda_1(S) \geq \psi_0) \to 0$, no bound on the convergence rate is available. 

The discrepancy induced by replacing $S$ with $\vS$ is negligible as indicated by the following lemma, proved in  Section S.9 of \citet{Li2020high} (see (S.9.1)). 
\begin{lemma}
    \label{lemma:sqrt_N_convergence_vS_S}
     Suppose \ref{enum:HD_regime}--\ref{enum:converge_X_Sigma_1} hold. As $N\to\infty$, $\sqrt{N} \|\vS - S \|_2 \stackrel{P}{\longrightarrow} 0$.
\end{lemma}
Also, since when $\lambda_1(S) \leq \psi_0$ and $\lambda_1(\vS)\leq \psi_0$, there exists a constant $\calK$ such that we have $\|(S - zI_N)^{-1}\|_2 \leq \calK$ and $\|(\vS - z I_N)^{-1}\|_2 \leq \calK$ for any $z\in\calC$. Therefore,
\[ \mP( \sup_{z\in\calC}\|(S - zI_N)^{-1}\|_2 \leq \calK ) \to 1,\]
\[ \mP( \sup_{z\in\calC}\|(\vS - zI_N)^{-1}\|_2 \leq \calK ) \to 1.\]

Using $\mu(z)$ as an example,  we define the truncated version of $\mu(z)$ as 
\[ \breve{\mu}(z) = \frac{1}{N} \tr[(\vS- zI_N)^{-1}], \quad z \in \calC.\]
Then, the discrepancy in the integral of the process induced by the variable truncation is negligible in the following sense     
\begin{align*} 
\oint_{\calC} |\breve{\mu}(z) - \mu(z)||dz| \leq \calK \| S -\vS \|_2 \sup_{z\in \calC}\| (\vS - z I_N)^{-1} \|_2  \| (S - z I_N)^{-1} \|_2 = O_p(\|S - \vS\|_2)  = o_p(N^{-1/2}). 
\end{align*}
Here, we are using the fact that 
\[ (\vS - z I_N)^{-1} - (S -zI_N)^{-1} = (\vS - z I_N)^{-1} ( S - \vS ) (S - z I_N)^{-1}.\]
Similarly, we obtain 
\begin{align*} 
\oint_{\calC} &\Big| \frac{1}{N} X^T (\vS -  zI_N)^{-1} X - \frac{1}{N} X^T (S - zI_N)^{-1} X\Big||dz|\\
&\leq \calK \|S - \vS\|_2 \frac{1}{N}\|X^T X\|_2 \sup_{z\in \calC}\| (\vS - z I_N)^{-1} \|_2  \| (S - z I_N)^{-1} \|_2  = O_p(\| S - \vS\|_2 ) = o_p(N^{-1/2}).
\end{align*}
Similar bounds can be obtained on the integrals of the difference of the other integrands induced by the variable truncation procedure. For this reason, to show the convergence in probability in Lemma \ref{lemma:convergence_integrals}, it suffices to consider the truncated sample covariance matrix $\vS$ instead of the original matrix $S$.

\emph{Process smoothing}. We demonstrate the smoothing procedure using $\breve{\mu}(z)$ as an example. The procedure can be readily applied to other processes. Select a sequence of $\rho_N >0$ such that for some $a \in (1,2)$
\[ N \rho_N \downarrow 0, \quad \rho_N \geq N^{-a}.\]
Let $\calC^+ = \calC \cap \{u+i v: |v|\geq \rho_N \}$.  Define a smoothed version of $\breve{\mu}(z)$ as
\[\breve{\underline{\mu}}(z) = \begin{cases} \breve{\mu}(z) & \mbox{ if }z\in \calC^+ \\ 
\cfrac{\rho_N - \Im(z) }{2\rho_N} \breve{\mu}(u+ i \rho_N) + \cfrac{\Im(z) +\rho_N}{2\rho_N} \breve{\mu}(u- i\rho_N), & \mbox{ if } z \in \calC \setminus\calC^+.
\end{cases}\]
It means that if $z$ is too close to the real axis, $\breve{\mu}(z)$ is modified to be the linear interpolation of its values at $u+i\rho_N$ and $u-i\rho_N$. 


We briefly summarize the purpose of the smoothing procedure. Again, we use $\breve{\mu}(z)$ as an example. The arguments are readily applied to other processes. First of all, 
\[ \oint_{\calC} |\underline{\breve{\mu}}(z) - \breve{\mu}(z)| |dz| \leq \calK \oint_{\calC\setminus\calC^+} |\underline{\breve{\mu}}(z) - \breve{\mu}(z)| |dz| \leq \calK \rho_N \sup_{z\in\calC} \| (\vS - z I_N)^{-1} \| = o_p(N^{-1/2}).\]
It indicates that the process smoothing procedure will only alter the integral negligibly in probability. Therefore, to show the convergence in probability in Lemma \ref{lemma:convergence_integrals}, it suffices to focus on the smoothed and truncated process $\underline{\breve{\mu}}(z)$ instead of the original process $\mu(z)$. In particular, we only need to show the following lemma.

\begin{lemma}
Suppose that \ref{enum:HD_regime}--\ref{enum:converge_X_Sigma_1} hold.  We have 
\begin{equation}\label{eq:converge_altered1}
    \sup_{z\in\calC^+} | \breve{\mu}(z) - s(z)| = o_p(N^{-1/2}).
\end{equation}
\begin{equation}\label{eq:converge_altered2}
    \sup_{z\in\calC^+} \Big\| X^T(\vS - zI_N^{-1})X  - \delta_\infty(z)\Big\|_2 = O_p(N^{-1/2}).
\end{equation}
\begin{equation}\label{eq:converge_altered3}
    \sup_{z\in\calC^+} \Big| \frac{\sigma_Z^2}{N}\tr[(\vS- zI_N)^{-1}\Sigma]  - \breve{\pi}(z) \Big| = o_p(N^{-1/2}).
\end{equation}
\begin{equation}\label{eq:converge_altered4}
    \sup_{z_1,z_2\in\calC^+}\Big|\frac{\sigma_Z^4}{N}\tr[(\vS- z_1 I_N)^{-1}\Sigma (\vS - z_2 I_N)^{-1}\Sigma] - \breve{\kappa}(z_1, z_2) \Big|= o_p(1),\\
\end{equation}
\begin{equation}\label{eq:converge_altered5}
\sup_{z_1,z_2\in\calC^+}\Big\|\frac{\sigma_Z^2}{N}X^T(\vS- z_1 I_N)^{-1}\Sigma (\vS - z_2 I_N)^{-1}X - \breve{\omega}(z_1, z_2) \Big\|_2 = o_p(1).    
\end{equation}
Here, $\breve{\pi}(z)$, $\breve{\kappa}(z_1, z_2)$, and $\breve{\omega}(z_1, z_2)$ are the truncated version of $\pi(z)$, $\kappa(z_1, z_2)$, and $\omega(z_1,z_2)$ after $S$ replaced by $\vS$, respectively.
\end{lemma}

Result \eqref{eq:converge_altered1} is proved in Sections 2, 3, and 4 of \citet{bai2008clt}. Result \eqref{eq:converge_altered2} can be proved by closely following their arguments. Indeed, instead of taking trace of $(\vS- zI_N)^{-1}$, we consider taking trace of $(\vS- zI_N)^{-1} a a^T$ for a general sequence of vectors $a$ such that $N^{-1} a^Ta < \infty$. Result \eqref{eq:converge_altered3} is 
shown in \citet{Li2020high} (see the proof of Lemma 2.2). 
We therefore omit the details. 

We only need to show \eqref{eq:converge_altered4} and \eqref{eq:converge_altered5}. The pointwise convergence is dealt with in Lemma \ref{lemma:converge_kappa_omega}. We only need to show the tightness on $\calC^+$.

First of all, on $\calC^+$, the spectral norm of the resolvent \( (\vS - z I_N)^{-1} \) is always bounded by \( 1/|\Im(z)|  = O(\rho_N^{-1}) \). If it is further known that the largest eigenvalue of $\vS$ is contained in $[0, \psi_0]$, we can find a constant $\calK$, such that the spectral norm of $(\vS- zI_N)^{-1}$ is bounded by $\calK$. Together, denote $\mathbb{G}$ to be the event $\{ \lambda_{\max}(\vS)> \psi_0\}$. For some constant $\calK$, 
\[ \sup_{z\in\calC^+} \| (\vS- z I_N)^{-1} \|_2 \leq   [ \calK \mathbbm{1}(\mathbb{G}^c) + \rho_N^{-1}\mathbbm{1}(\mathbb{G})] \leq \calK[1+ \rho_N^{-1}\mathbbm{1}(\mathbb{G})]. \]
Lemma \ref{lemma:poly_bound_prob_exceed_right} says, for any positive $l$,
\[\mP(\mathbb{G}) = o(N^{-l}).\]
Therefore,  we achieve the uniform boundedness of the moments of the spectral norm as 
\begin{equation}\label{eq:tightness_e1}
\sup_{z\in\calC^+}\mE \|(\vS - zI_N)^{-1} \|^k_2 \leq \calK^k [1+ \rho_N^{-1} \mP(\mathbb{G})]^k < \calK, \quad \mbox{ for any } k\geq 2.
\end{equation}
Therefore,
\[\sup_{z_1, z_2\in \calC^+} \Big|\frac{\sigma_Z^4}{N} \tr[ (\vS - z_1 I_N)^{-1} \Sigma (\vS - z_2 I_N)^{-1}\Sigma] \Big| \leq \calK.\]
It implies that the process is tight. Similarly, we can show $\breve{\kappa}(z_1,z_2)$, $N^{-1} X^T (\vS - z_1 I_N)^{-1} \Sigma (\vS - z_2 I_N)^{-1} X $, and $\breve{\omega}(z_1, z_2)$ are tight.

\section{Proof of Theorem \ref{thm:normality}}\label{sec:proof_normality}
Let $\hat{\beta}$ be the TLS estimator with the weight matrix $f(S)$ and the true inflation factor $\sigma_X^2$. In this section, we show the asymptotic normality of $\hat\beta$, following the strategy of \citet{li2023regularized}. We shall consider the case when $D$ is fixed. The proof when $\|D\|_2\to 0$, as $N\to\infty$, can be obtained by modifying the following arguments. Under the framework \eqref{eq:fingerprint1}--\eqref{eq:fingerprint3}, we can express $\tilde{X}$ as 
\[ \tilde{X} = X + \sigma_X V D^{1/2},\]
where $V = [V_1, V_2, \dots, V_p] \stackrel{iid}{\sim} \calN(0,\Sigma)$ and $V$ is independent of $\epsilon$. Consider the eigen-decomposition of $[g(S)]^{1/2} \Sigma [g(S)]^{1/2} = U \Lambda U^T$, where $\Lambda = \operatorname{diag}(\alpha_i)$ is the diagonal matrix of eigenvalues, and $U$ is the corresponding matrix of eigenvectors. Let 
\begin{equation}\label{eq:def_stars}
\begin{aligned}
&Y^* = U^T [g(S)]^{1/2} Y, \quad  X^* =  \sigma_X^{-1} U^T [g(S)]^{1/2} X D^{-1/2}, \quad \epsilon^* = U^T [g(S)]^{1/2} \epsilon, \\
&\tilde{X}^* = \sigma^{-1}_X U^T [g(S)]^{1/2} \tilde{X} D^{-1/2}, \quad V^* = U^T [g(S)]^{1/2} V, \quad \beta^* = \sigma_X D^{1/2} \beta.
\end{aligned}
\end{equation}
With the definition, we have
\[ Y^* = X^* \beta^* + \epsilon^*,\]
\[  \tilde{X}^* = X^* + V^*. \]
\[ \frac{\|[g(S)]^{1/2}(Y - \tilde{X}\beta )\|^2_2}{1 + \sigma^2_X \beta^T D\beta } = \frac{\|Y^* - \tilde{X}^*\beta^* \|^2_2}{1+{\beta^*}{}^T\beta}.\]
Clearly, if $\hat{\beta}^*$ is the minimizer of the right-hand side, the TLS estimator is $\hat{\beta} = \sigma_X^{-1}  D^{-1/2}\hat{\beta}^*$. 

Taking the derivative of the objective function with respect to $\beta^*$, we find that the minimizer solves the score equation 
\[ S(\hat{\beta}^*) = \frac{{\tilde{X}^*}{}^T (Y^* - \tilde{X}^* \hat{\beta}^* ) }{N} + \hat{\beta}^* \frac{\| Y^* - \tilde{X}^*\hat{\beta}^* \|_F^2 }{N( 1+ {\hat{\beta}^*}{}^T \hat{\beta}^* )}  = 0,\]
where $S(\beta) = (S_1(\beta), S_2(\beta), \dots, S_p(\beta))^T$ is a $p$-dimensional vector. By Taylor's theorem, there exists a series of $\tilde{\beta}^*_j$ on the line segment between $\hat{\beta}^*_j$ and $\beta^*_j$ for $j=1,2,\dots, p$, such that 
\[ S(\hat{\beta}^*) = S(\beta^*) + H  (\hat{\beta}^* - \beta^*) = 0,\]
where $H = \nabla S(\tilde{\beta}^*) $ is the $p\times p$ Hessian matrix at $\tilde{\beta}^* = (\tilde{\beta}^*_1, \tilde{\beta}^*_2, \dots, \tilde{\beta}^*_p)^T$.

It follows that 
\begin{align*}
H \sqrt{N} (\hat{\beta}^* - \beta^*) &= -\sqrt{N} S(\beta^*)\\
&= -\frac{1}{\sqrt{N}} \sum_{i=1}^N\Big\{  (\epsilon_i^* - {v_i^*}{}^T\beta^*)( x_i^* + v_i^*) + \beta^* \frac{ (\epsilon^*_i - {v_i^*}{}^T \beta^*)^2 }{1+ {\beta^*}{}^T\beta^* } \Big\} \\
& \coloneqq -\frac{1}{\sqrt{N}} \sum_{i=1}^N M_i(\beta^*), 
\end{align*}
where  for $i=1,2,\dots, N$, $\epsilon_i^*$ is the $i$th element of $\epsilon^*$, $v_i^*$ is the $i$th row vector of the matrix $V^*$, and $x_i^*$ is the $i$th row vector of the matrix $X^*$.  

From the normality of $\epsilon$ and $V$, it is straightforward that $\epsilon_i^* \mid S \sim \calN(0, \alpha_i)$, $v_i^* \mid S  \sim \calN(0, \alpha_i I_p)$, and $M_i(\beta^*)\mid S$ are mutually independent vectors with finite (conditional) covariance matrices: 
\begin{align*}
    \operatorname{cov}(M_i(\beta^*) \mid S)  & = \alpha_i x_i^* {x_i^*}{}^T ( 1+ {\beta^*}{}^T \beta^*) + \alpha_i^2 ( I_p  + I_p {\beta^*}{}^T \beta^* + 2 \beta^*{\beta^*}{}^T) - 3 \alpha_i^2 \beta^*{\beta^*}{}^T\\
    & = [\alpha_i x_i^* {x_i^*}{}^T + \alpha_i^2 (I_p  + \beta^* {\beta^*}{}^T )^{-1}] (1+{\beta^*}{}^T \beta^*).
\end{align*}
It follows that 
\begin{align*}
   \mathfrak{M} \coloneqq \frac{1}{N} \sum_{i=1}^N \operatorname{cov}(M_i(\beta^*)\mid S)  = \Big\{\frac{1}{N} \sigma_X^{-2}D^{-1/2} X^T g(S) \Sigma g(S) XD^{-1/2} \\
   + \frac{1}{N}\tr[ g(S) \Sigma g(S) \Sigma ] (I_p  + \sigma^2_X D^{1/2}\beta {\beta}^TD^{1/2} )^{-1} \Big\} (1+ \sigma_X^2 {\beta}^T D\beta).  
 \end{align*}
Moreover, it is straightforward to show that $\mE (M_i(\beta^*)\mid S) = 0$. 

Using the Lindeberg--Feller central limit theorem, we obtain that for any vector $a \in \mathbb{R}^p$ and any $t \in \mathbb{R}$,
\[
\mathbb{P} \left(
    \frac{1}{\sqrt{N\, a^\top \mathfrak{M} a}} 
    \sum_{i=1}^N a^\top M_i(\beta^*) 
    \leq t
    \;\middle|\; \lambda_1(S) \leq \psi_0
\right) \longrightarrow \Phi(t), \quad \mbox{as }N\to\infty,
\]
where $\Phi(\cdot)$ denotes the cumulative distribution function of the standard normal distribution $\calN(0,1)$. Here, the Lindeberg condition is satisfied due to the existence of higher-order moments of $\epsilon_i\mid S$ and $v_i^* \mid S$ and the boundedness of $\mathfrak{M}$ when $\lambda_1(S)$ is contained within $[0, \psi_0]$.  

Since $\mP(\lambda_1(S) \leq \psi_0 ) \to 1$,  we conclude that 
\[ \frac{1}{\sqrt{N\, a^\top \mathfrak{M} a}} 
    \sum_{i=1}^N a^\top M_i(\beta^*)  \stackrel{D}{\longrightarrow } \calN(0,1).\]

Moreover, in Section \ref{sec:convergence_integrals}, we proved that 
\begin{align*}
\mathfrak{M} &= \Big\{ \sigma_X^{-2} D^{-1/2} \Omega(f) D^{-1/2} +  K(f) [ I_p + \sigma_X^2D^{1/2} \beta{\beta}^TD^{1/2}]^{-1}\Big\}(1+\sigma_X^2{\beta}^TD\beta) + o_p(1)\\
&= \Big\{ \sigma_X^{-2} D^{-1/2} \Omega_\infty(f) D^{-1/2} +  K_\infty(f) [ I_p + \sigma_X^2D^{1/2} \beta{\beta}^TD^{1/2}]^{-1}\Big\}(1+\sigma_X^2{\beta}^TD\beta) + o_p(1).
\end{align*}
Next, consider the Hessian at $\tilde{\beta}^*$.  At any value $\beta_0$,
\begin{align*}
\nabla S(\beta_0) &= -\frac{{\tilde{X}^*}{}^T\tilde{X}^*}{N} + \frac{\| Y^* - \tilde{X}^*\beta_0\|_2^2 }{N(1+ \beta_0^T\beta_0) } I_p + \beta_0 \frac{\partial \|Y^* - \tilde{X}^* \beta_0 \|_F^2 / \{ N(1+\beta_0^T \beta_0)\}}{\partial \beta^T_0}\\
& = -\frac{{\tilde{X}^*}{}^T\tilde{X}^*}{N} + \frac{\| Y^* - \tilde{X}^*\beta_0\|_2^2 }{N(1+ \beta_0^T\beta_0) } I_p + \frac{\beta^T_0\{S(\beta_0)\}^T}{ 1+ \beta^T_0 \beta_0}.
\end{align*}
Since as $N\to\infty$ and $\beta_0 \longrightarrow \hat{\beta}^*$, we have $S(\beta_0) \to 0$, we have  the Hessian at $\tilde{\beta}^*$ is such that 
\begin{align*}
H  &= -\frac{{\tilde{X}^*}{}^T \tilde{X} }{N} +   \frac{\|Y^* - \tilde{X}^* \tilde{\beta}^*\|_F^2 }{N(1+{\tilde{\beta}^*}{}^T \tilde{\beta}^*) } I_p +  o_p(1)\\    
& = -\frac{1}{N}\sigma_X^{-2} D^{-1/2} X^T g(S) \Sigma g(S) X D^{-1/2} - \frac{1}{N}\tr(g(S)\Sigma) I_p + \frac{1}{N} \tr(g(S)\Sigma)I_p \\
& \quad \quad + \frac{(\tilde{\beta}^* - \beta^*)^T \frac{1}{N}{X^*}{}^T g(S) \Sigma g(S) X^* (\tilde{\beta}^* - \beta^*) }{1 + {\tilde{\beta}^*}{}^T \tilde{\beta}^*} I_p + o_p(1)\\
& = -\frac{1}{N}\sigma_X^{-2} D^{-1/2} X^T g(S) \Sigma g(S) X D^{-1/2} + o_p(1)\\
& = - \sigma_X^{-2} D^{-1/2}\Delta(f) D^{-1/2} + o_p(1).
\end{align*}
In summary, due to Slutsky's theorem, we conclude that 
\[ \sqrt{N}  \Big\{\sigma_X^4 D^{1/2} (\Delta(f))^{-1} D^{1/2} \mathfrak{M} D^{1/2} (\Delta(f))^{-1} D^{1/2} \Big\}^{-1/2}  (\hat{\beta}^* -\beta^*) \stackrel{D}{\longrightarrow} \calN(0, I_p).\]
Since $\hat{\beta} = \sigma_X^{-1} D^{-1/2}\hat{\beta}^*$, we obtain that 
\[ \sqrt{N} [\Xi(f)]^{-1/2} (\hat{\beta}-\beta) \stackrel{D}{\longrightarrow} \calN(0, I_p).\]

\section{Proof of Technical Results in Section \ref{subsec:estimation_a}}\label{sec:proof_inflation_factor_estimator}

\subsection{Proof of Theorem \ref{thm:asymptotic_VX_VZ} and Corollary \ref{corollary:consistency_sigma_hat_Z_known}}\label{subsec:proof_asymptotic_VX_VZ}
The asymptotic normality of $V_X$ and $V_Z$ in Theorem \ref{thm:asymptotic_VX_VZ} follows directly from Lindeberg-Feller's central limit theorem. As for Corollary \ref{corollary:consistency_sigma_hat_Z_known}, it is routine to verify the convergence of $\hat{H}^2$ since $m S \sim \operatorname{Wishart}(m, \sigma_Z^2\Sigma)$. The convergence of $\widetilde{\sigma}_X^2$ follows immediately from Slutsky's theorem and Theorem \ref{thm:asymptotic_VX_VZ}.

\subsection{Proof of Theorem \ref{thm:asymptotic_normality_estimated_sigma_Z_known}}\label{subsec:proof_asymptotic_normality_estimated_sigma_Z_known}

Recall the definitions in \eqref{eq:def_stars}. Again, similar to the proof of \ref{thm:normality}, we shall only consider the case when $D$ is fixed. The case when $\|D\|_2 \to 0$ can be proved with modifications to the following arguments. 

The objective function with the estimated inflation factor $\widetilde{\sigma}_X^2$ and the weight matrix $f(S)$ is such that 
\[ \frac{\|[g(S)]^{1/2}(Y - \tilde{X}\beta )\|^2_2}{1 + \widetilde{\sigma}^2_X \beta^T D\beta } = \frac{\|Y^* - \tilde{X}^*\beta^* \|^2_2}{1+ (\widetilde{\sigma}_X^2/\sigma_X^2) {\beta^*}{}^T\beta^*}.\]
Clearly, if $\hat{\beta}^*$ is the minimizer of the right-hand side, $\hat{\beta}(\widetilde{\sigma}_X^2, f(S) ) = \sigma_X^{-1}  D^{-1/2}\hat{\beta}^*$. 

Taking the derivative of the objective function with respect to $\beta^*$, we find that the minimizer solves the score equation 
\[ \tilde{S}(\hat{\beta}^*) = \frac{{\tilde{X}^*}{}^T (Y^* - \tilde{X}^* \hat{\beta}^* ) }{N} + (\widetilde{\sigma}_X^2/\sigma_X^2) \hat{\beta}^* \frac{\| Y^* - \tilde{X}^*\hat{\beta}^* \|_F^2 }{N( 1+ (\widetilde{\sigma}_X^2/\sigma_X^2){\hat{\beta}^*}{}^T \hat{\beta}^* )}  = 0,\]
where $\tilde{S}(\beta) = (\tilde{S}_1(\beta), \tilde{S}_2(\beta), \dots, \tilde{S}_p(\beta))^T$ is a p-dimensional vector. By Taylor's theorem, there exists a series of $\tilde{\beta}^*_j$ on the line segment between $\hat{\beta}^*_j$ and $\beta^*_j$ for $j=1,2,\dots, p$, such that 
\[ \tilde{S}(\hat{\beta}^*) = \tilde{S}(\beta^*) + \tilde{H}  (\hat{\beta}^* - \beta^*) = 0,\]
where $\tilde{H} = \nabla \tilde{S}(\tilde{\beta}^*) $ is the $p\times p$ Hessian matrix at $\tilde{\beta}^* = (\tilde{\beta}^*_1, \tilde{\beta}^*_2, \dots, \tilde{\beta}^*_p)^T$.

It follows that 
\begin{align*}
\tilde{H} \sqrt{N} (\hat{\beta}^* - \beta^*) &= -\sqrt{N} \tilde{S}(\beta^*)\\
&= -\frac{1}{\sqrt{N}} \sum_{i=1}^N\Big\{  (\epsilon_i^* - {v_i^*}{}^T\beta^*)( x_i^* + v_i^*) + \beta^* \frac{ (\epsilon^*_i - {v_i^*}{}^T \beta^*)^2 }{1+ {\beta^*}{}^T\beta^* } \Big\} \\
&\quad + \Big[\frac{(V_X/\bar{\tau} )\beta^*}{ 1 + (V_X/\bar{\tau}) {\beta^*}{}^T \beta^* } - \frac{\beta^*}{ 1+ {\beta^*}{}^T \beta^* } \Big] \frac{1}{\sqrt{N}}\sum_{i=1}^N (\epsilon^*_i - v_i^* \beta^*)^2  \\
&\quad + \Big[\frac{(\widetilde{\sigma}_X^2/\sigma_X^2)\beta^*}{ 1 + (\widetilde{\sigma}_X^2/\sigma_X^2) {\beta^*}{}^T \beta^* } - \frac{(V_X/\bar{\tau})\beta^*}{ 1+ (V_X/\bar{\tau}) {\beta^*}{}^T \beta^* } \Big] \frac{1}{\sqrt{N}}\sum_{i=1}^N (\epsilon^*_i - v_i^* \beta^*)^2  \\
& \coloneqq \frac{1}{\sqrt{N}} \sum_{i=1}^N Q^{(1)}_i(\beta^*) + \frac{1}{\sqrt{N}} \sum_{i=1}^N Q_i^{(2)}(\beta^*)  + \frac{1}{\sqrt{N}} \sum_{i=1}^N Q_i^{(3)}(\beta^*), 
\end{align*}
where $\bar{\tau} = \sigma_X^2 M_1$.

The asymptotic normality of $N^{-1/2} \sum_{i=1}^N Q_i^{(1)} (\beta^*)$ is established in Section \ref{sec:proof_normality}. Indeed, we have 
\[ \frac{1}{\sqrt{N}} [\mathfrak{M}]^{-1/2} \sum_{i=1}^N Q_i^{(1)} (\beta^*)  \stackrel{D}{\longrightarrow}\calN(0, I_p).\]
See Section \ref{sec:proof_normality} for the definition of $\mathfrak{M}$.

As for the third term, following immediately from Theorem \ref{thm:asymptotic_VX_VZ},
\[ \widetilde{\sigma}_X^2/ \sigma_X^2 - V_X/\bar{\tau} = V_X \Big[ \frac{\sigma_Z^2 N^{-1}\tr(\Sigma) - V_Z}{V_Z\bar{\tau}}\Big] = O_p(N^{-1/2}m^{-1/2}) = O_p (N^{-1}).\]
Therefore, we conclude that 
\[   \frac{1}{\sqrt{N}} \sum_{i=1}^N Q_i^{(3)}(\beta^*) = O_p(N^{-1/2}).\]
As for the second term,
\begin{align*}
\frac{1}{\sqrt{N}} \sum_{i=1}^N Q_i^{(2)}(\beta^*) &= \sqrt{N}\Big[\frac{(V_X/\bar{\tau} )\beta^*}{ 1 + (V_X/\bar{\tau}) {\beta^*}{}^T \beta^* } - \frac{\beta^*}{ 1+ {\beta^*}{}^T \beta^* } \Big]  \frac{1}{N}\sum_{i=1}^N \mE (\epsilon^*_i - v_i^* \beta^*)^2 \\
&\quad +  \sqrt{N}\Big[\frac{(V_X/\bar{\tau} )\beta^*}{ 1 + (V_X/\bar{\tau}) {\beta^*}{}^T \beta^* } - \frac{\beta^*}{ 1+ {\beta^*}{}^T \beta^* } \Big] \Big[\frac{1}{N}\sum_{i=1}^N (\epsilon^*_i - v_i^* \beta^*)^2 - \frac{1}{N}\sum_{i=1}^N \mE (\epsilon^*_i - v_i^* \beta^*)^2 \Big]\\
&= \sqrt{N}\Big[\frac{(V_X/\bar{\tau} )\beta^*}{ 1 + (V_X/\bar{\tau}) {\beta^*}{}^T \beta^* } - \frac{\beta^*}{ 1+ {\beta^*}{}^T \beta^* } \Big]   (1 + {\beta^*}{}^T \beta^*) \frac{1}{N}\tr(g(S)\Sigma) + O_p\Big(\frac{1}{\sqrt{n_T - p} \sqrt{N}}\Big).
\end{align*}
In summary, by the delta method, we get
\begin{align*}
    \sqrt{n_T - p} \frac{M_1}{M_2 \Pi(f)} \frac{1}{\sqrt{N}} \sum_{i=1}^N Q_i^{(2)}(\beta^*)  \stackrel{D}{\longrightarrow} \calN( 0,  \beta^*{\beta^*}{}^T).
\end{align*}

As for the covariance between the first and the second term, note that $V_X$ is independent of $Y$ and $\tilde{X}$ due to the normality assumption of the noise. It is because $V_X$ is determined by the within-group residual covariance matrix and $\tilde{X}$ is determined by the group means.  We conclude that 
\[\sqrt{N}\Big[\frac{(V_X/\bar{\tau} )\beta^*}{ 1 + (V_X/\bar{\tau}) {\beta^*}{}^T \beta^* } - \frac{\beta^*}{ 1+ {\beta^*}{}^T \beta^* } \Big] 
\mbox{ is independent of } 
\frac{1}{\sqrt{N}} \sum_{i=1}^N Q_i^{(1)}(\beta^*).\]

Moreover, following analogous arguments as in Section \ref{sec:proof_normality}, we can show that 
\[ \tilde{H} = - \sigma_X^{-2} D^{-1/2} \Delta(f) D^{-1/2} + o_p(1).\]
Let
\[ \Gamma(f) = \frac{\sigma_X^4 M_2^2\Pi^2(f)}{M_1^2 (n_T - p)} \Delta^{-1}D\beta \beta^T D \Delta^{-1}.\]
In conclusion, 
\[ \sqrt{N} [ \Xi(f) + \Gamma(f) ]^{-1/2} \Big(\hat{\beta}(\widetilde{\sigma}_X^2) - \beta\Big)  \stackrel{D}{\longrightarrow} \calN(0, I_p). \]

\subsection{Proof of Theorem \ref{thm:consistency_a_b}}
Let $\hat{\beta}(\sigma^2)$ be the TLS estimator with the weight matrix $f(S)$ and a general inflation factor $\sigma^2$. Using Lemma \ref{lemma:explicit_formula_of_beta},
\[\hat\beta(\sigma^2) = \Big[\frac{1}{N}\tilde{X}^T g(S) \tilde{X} - \lambda(\sigma^2)D\Big]^{-1}\frac{1}{N} \tilde{X}^T g(S) Y,\]
where $\lambda(\sigma^2)$ is the smallest eigenvalue of the matrix 
\[  A(\sigma^2) =  \frac{1}{N} \left[\begin{matrix} D^{-1/2} & &0\\ 0& & \sigma \end{matrix} \right]\left[ \begin{matrix} \tilde{X}^T \\ Y^T  \end{matrix} \right] g(S) \left[\begin{matrix} \tilde{X}, &\, Y  \end{matrix} \right] \left[\begin{matrix} D^{-1/2} && 0\\ 0 && \sigma \end{matrix} \right]. \]
We first study the behavior of $A(\sigma^2)$. Notice that we can express $\tilde{X}$ and $Y$ as $\tilde{X} = X  + \sigma_X V D^{1/2}$ and $Y = X \beta + \epsilon$, where $V = [V_1, \dots ,V_p] \stackrel{iid}{\sim} \calN(0,\Sigma)$, $\epsilon \sim \calN(0,\Sigma)$, and $V$ is independent of $\epsilon$.  
\begin{lemma}\label{lemma:concentration_tildeX_fS_tildeX}
Suppose that \ref{enum:HD_regime}--\ref{enum:converge_X_Sigma_1} hold. Then,
\[ \frac{1}{N} X^T g(S) X = \Delta(f) + O_p(N^{-1/2}) = \Delta_\infty(f) + O_p(N^{-1/2}),\]  
\[\frac{1}{N} \sigma_X^2 V^T g(S) V = \sigma^2_X \Pi(f) I_p + O_p(N^{-1/2})=\sigma^2_X \Pi_\infty(f) I_p + O_p(N^{-1/2}),\]
\[ \frac{1}{N}\sigma_X X^T g(S) V = O_p(N^{-1/2}), \quad \frac{1}{N} \epsilon^T g(S) \epsilon =  \Pi(f) + O_p(N^{-1/2})=  \Pi_\infty(f) + O_p(N^{-1/2}),\]
\[ \frac{1}{N} X^T g(S) \epsilon = O_p(N^{-1/2}), \quad \frac{1}{N} \sigma_X V^T g(S) \epsilon = O_p(N^{-1/2}).\]
\end{lemma}
\noindent The lemma follows directly from Lemma \ref{lemma:concentration_quadratic_forms} and Lemma \ref{lemma:convergence_Delta_Pi_Omega_K}. The proof is omitted. 

Following Lemma \ref{lemma:concentration_tildeX_fS_tildeX}, we can conclude that 
\begin{align*}
    A(\sigma^2) & = \left[\begin{matrix} -D^{-1/2} \\  \sigma  \beta^T \end{matrix} \right] {\Delta}(f)  \left[\begin{matrix}  -D^{-1/2}, \sigma \beta \end{matrix} \right] + \left[\begin{matrix}
        \sigma_X^2 \Pi(f) I_p & 0\\ 0 &\sigma^2\Pi(f)
    \end{matrix}\right]  + Q_A(\sigma^2) \\
    & = \left[\begin{matrix} -D^{-1/2} \\  \sigma  \beta^T \end{matrix} \right] {\Delta}_\infty(f)  \left[\begin{matrix}  -D^{-1/2}, \sigma \beta \end{matrix} \right] + \left[\begin{matrix}
        \sigma_X^2 \Pi_\infty(f) I_p & 0\\ 0 &\sigma^2\Pi_\infty(f)
    \end{matrix}\right]   + Q_A(\sigma^2)\\
    & \coloneqq A_\infty(\sigma^2)  + Q_A(\sigma^2), \quad \mbox{say},
\end{align*}
where $Q_A(\sigma^2)$ is a residual term such that 
\[ \sup_{\sigma^2\in [0, \Upsilon]} \|Q_A(\sigma^2)\|_2  = O_p\Big( \frac{n_T}{\sqrt{N}} \Big).\]
Denote the smallest eigenvalue of $A_\infty(\sigma^2)$ to be $\lambda_\infty(\sigma^2)$. It is straightforward to verify that when $N$ is sufficiently large, for any $\sigma^2\in [0, \Upsilon]$,  
\[ \lambda_\infty(\sigma^2) = \sigma^2 \Pi_\infty(f). \]
Therefore, we have, uniformly on $[0, \Upsilon]$,
\[ \lambda(\sigma^2) - \lambda_\infty(\sigma^2)  = O_p\Big(\frac{n_T}{\sqrt{N}} \Big).\]
Summarizing these results, we conclude that when $N$ is sufficiently large, uniformly on $[0, \Upsilon]$,
\[ \left[\frac{1}{N}\tilde{X}^T g(S) \tilde{X} - \lambda(\sigma^2) D \right]^{-1}  = \Big[ \Delta_\infty(f) + (\sigma^2 - \sigma_X^2)\Pi_\infty(f) D \Big]^{-1} + O_p\Big(\frac{1}{\sqrt{N}} \Big). \]

On the other hand, consider the behavior of $\frac{1}{N} \tilde{X}^T g(S) Y$. 
\begin{align*}
\frac{1}{N}\tilde{X}^T g(S) Y = &\frac{1}{N} (X^T g(S) X) \beta + \frac{1}{N} X^T g(S)\epsilon + \frac{\sigma_X}{N}  D^{1/2} V^T g(S) X \beta  + \frac{\sigma_X}{N} D^{1/2} V^T g(S) \epsilon\\
 = &\frac{1}{N} (X^T g(S) X) \beta + \frac{1}{N} X^T g(S)\Sigma^{1/2}\zeta + \frac{\sigma_X}{N}  D^{1/2} W^T\Sigma^{1/2} g(S) X \beta  + \frac{\sigma_X}{N} D^{1/2} W^T \Sigma^{1/2} g(S)\Sigma^{1/2} \zeta,
\end{align*}
where $\epsilon = \Sigma^{1/2}\zeta$ and $V = \Sigma^{1/2} W$.

The first term is such that 
\[ \frac{1}{N} (X^T g(S) X) \beta  = \Delta_\infty(f) \beta + O_p(1/\sqrt{N}).\]

Following the strategy in  Section \ref{sec:proof_normality}, we can write the other terms as the average of independent terms of finite variance. Indeed, we consider the eigen-decomposition of $ \Sigma^{1/2} [g(S)]\Sigma^{1/2} = Q \Lambda Q^T$, where 
\[\Lambda = \operatorname{Diag}(\alpha_1,\dots, \alpha_N).\] 
We rotate $\zeta$ and $W$ by $Q$ as $W^* = Q^T W$ and $\zeta^* = Q^T \zeta$. Due to the normality assumption, the entries of $W^*$ and $\zeta^*$ are still iid $\calN(0,1)$. Also, we rotate $\Sigma^{-1/2} X$ to be $X^* = Q^T \Sigma^{-1/2}X$. 
These terms are such that 
\begin{align*}
&\frac{1}{N} X^T g(S)\Sigma^{1/2} \zeta   + \frac{\sigma_X}{N}  D^{1/2} W^T \Sigma^{1/2}g(S) X \beta  + \frac{\sigma_X}{N} D^{1/2}W^T \Sigma^{1/2}g(S)\Sigma^{1/2}\zeta\\
= &\frac{1}{N} \sum_{j=1}^N  \Big\{ \lambda_j \zeta_j^* {X^*_{j\cdot}}^T  + \sigma_X\lambda_j ({X^*_{j\cdot}} \beta ) D^{1/2} {W^*_{j\cdot}}^T + \sigma_X\lambda_j  \zeta^*_j D^{1/2}{W^*_{j\cdot}}^T\Big\}
\end{align*}
where $\zeta^*_j$ is the $j$th element of $\zeta^*$ and  $X^*_{j\cdot}$ is the $j$th row of $X^*$. Notice that these terms are independent random vectors of mean $0$ and uniformly bounded finite variances.  

It follows then 
\[ \frac{1}{N} \tilde{X}^T g(S)Y = \Delta_\infty(f)\beta + O_p(N^{-1/2}).\]




Summarizing the results, we obtain the following lemma.

\begin{lemma}
    \label{lemma:convergence_beta_sigma}
    Suppose that \ref{enum:HD_regime}--\ref{enum:converge_X_Sigma_1} hold. For any $\sigma^2 \in [0, \Upsilon]$, define 
    \[ \beta(\sigma^2) = \Big[\Delta_\infty(f) - (\sigma^2 - \sigma_X^2) \Pi_\infty(f) D\Big]^{-1} \Delta_\infty(f) \beta .\]
    We have 
    \[ \sup_{\sigma^2 \in [0,\Upsilon]} \|\hat{\beta}(\sigma^2)  -\beta(\sigma^2)\|_2  = O_p\Big(\frac{1}{\sqrt{N}}\Big). \]
\end{lemma}

To show the convergence of $\widehat{\sigma}_X^2$, consider the loss function $L(\sigma^2)$. First of all, under the assumptions of Theorem \ref{thm:consistency_a_b}, 
\[  V_X - \sigma_X^2N^{-1}\operatorname{tr}(\Sigma) = o_p(N^{-1/2}).\]
Consider $\|Y - \tilde{X}\hat{\beta}(\sigma^2) \|_2^2$. We decompose it as 
\begin{align*}
Y - \tilde{X}\hat{\beta}(\sigma^2) = & \Big[ Y - \tilde{X}\beta(\sigma^2)\Big] + \Big[\tilde{X}\beta(\sigma^2)  -\tilde{X} \hat{\beta}(\sigma^2)\Big]\\
= & \Big[\epsilon  - \sigma_X V D^{1/2} \beta(\sigma^2)\Big] + X\Big[\beta - \beta(\sigma^2)\Big] +  \Big[\tilde{X}\beta(\sigma^2)  -\tilde{X} \hat{\beta}(\sigma^2)\Big].
\end{align*}
Using Lemma \ref{lemma:concentration_quadratic_forms}, we obtain that uniformly on $\sigma^2 \in [0, \Upsilon]$,
\[\frac{1}{N}\Big\|\epsilon  - \sigma_X V D^{1/2} \beta(\sigma^2)\Big\|_2^2 = (1 + \sigma_X^2 \beta^T(\sigma^2) D\beta(\sigma^2) )\frac{1}{N}\tr(\Sigma) + O_p(N^{-1/2}).\]
Using Lemma \ref{lemma:convergence_beta_sigma}, we obtain that uniformly on $\sigma^2 \in [0, \Upsilon]$,
\[ \frac{1}{N} \Big\|\tilde{X} [\beta(\sigma^2) - \hat{\beta}(\sigma^2)]  \Big\|_2^2 = O_p(N^{-1}).\]
Together, we conclude that uniformly on $[0, \Upsilon]$,
\[ \frac{1}{N}\Big\|Y - \tilde{X} \hat{\beta}(\sigma^2) \Big\|_2^2   =   (1 + \sigma_X^2 \beta^T(\sigma^2) D\beta(\sigma^2) )\frac{1}{N}\tr(\Sigma) + \frac{1}{N}[\beta(\sigma^2) -\beta]^T  X^TX[\beta(\sigma^2) -\beta] + O_p(N^{-1/2}).\]

Define 
\[ h(\sigma^2) = \frac{\sigma^2 \Big[ (1+ \sigma_X^2 \beta^T(\sigma^2) D\beta(\sigma^2))\frac{1}{N}\tr(\Sigma) + \frac{1}{N}[\beta(\sigma^2) - \beta]^T X^T X [\beta(\sigma^2) - \beta] \Big]}{1+ \sigma^2 \beta^T(\sigma^2) D\beta(\sigma^2)} - \sigma^2_X \frac{1}{N}\tr(\Sigma) \]
\[ h_n(\sigma^2) = \frac{\sigma^2\|Y -\tilde{X}\hat{\beta}(\sigma^2)\|_2^2}{N (1+ \sigma^2\hat{\beta}^T(\sigma^2) D \hat{\beta}(\sigma^2))} - V_X.\]
We conclude that uniformly on $\sigma^2\in[0,\Upsilon]$,
\[  h_n(\sigma^2) - h(\sigma^2) \stackrel{P}{\longrightarrow} 0.\]
Moreover, at $\sigma^2 = \sigma_X^2$, it is easy to verify that $h(\sigma^2_X) = 0$. For all sufficiently large $N$, $\sigma^2$ is the unique root to $h(\sigma^2) = 0$. Therefore, let $\widehat{\sigma}_X^2$ be a minimizer of $|h_n(\sigma^2)|$ on $[0, \Upsilon]$. We can conclude that 
\[ \widehat{\sigma}_X^2 - \sigma^2_X \stackrel{P}{\longrightarrow}0.\]
Moreover, it is straightforward to verify that there exists a neighborhood of $\sigma^2_X$ such that within the neighborhood, $|h'(\sigma^2)| > \calK >0$. 

On the other hand, since $h(\Upsilon) > 0$ for all sufficiently large $N$ when $\Upsilon>\sigma_X^2$, we conclude that 
\[ \mP\Big(\frac{\Upsilon \|Y -\tilde{X}\hat{\beta}(\Upsilon)\|^2_2}{N(1+\Upsilon \hat\beta^T(\Upsilon) D\hat{\beta}(\Upsilon)  ) }  > V_X  \Big) \longrightarrow 1, \quad \mbox{as }N\to\infty. \]
Therefore, with high probability, the solution $\widehat{\sigma}_X^2$ on $[0, \Upsilon]$ is such that 
\[ h_n(\widehat{\sigma}_X^2) = 0.\]
It follows then
\[ h(\widehat{\sigma}_X^2) - h(\sigma_X^2) =   h'(s) (\widehat{\sigma}^2_X - \sigma_X^2) =  h_n(\widehat{\sigma}_X^2) - h({\sigma}_X^2) + O_p(N^{-1/2}) = O_p(N^{-1/2}), \]
for some $s$ in between $\widehat{\sigma}_X^2$ and $\sigma_X^2$. We conclude then
\[(\widehat{\sigma}^2_X - \sigma_X^2) = O_p(N^{-1/2}).\]

Next, we consider the asymptotic distribution of $\hat{\beta}(\widehat{\sigma}_X^2)$. Since $V_Z/V_X = \sigma_Z^2/ \sigma_X^2 + O_p(N^{-1/2})$, we have that 
\[ \widehat{\sigma}_Z^2 = \widehat{\sigma}_X^2 V_Z/V_X = \sigma_Z^2 + O_p(N^{-1/2}).\]

Using Lemma \ref{lemma:explicit_formula_of_beta}, the difference between $\hat{\beta}(\widehat{\sigma}_X^2)$ and $\hat{\beta}(\sigma_X^2)$ is such that 
\[\sqrt{N} [\hat{\beta}(\widehat{\sigma}_X^2) - \hat{\beta}(\sigma_X^2)] = \sqrt{N}[\lambda(\widehat{\sigma}_X^2) - \lambda(\sigma^2)] \Big[ \frac{1}{N} \tilde{X}^T g(S) \tilde{X} - \lambda(\widehat{\sigma}^2_X) D\Big]^{-1}  D \Big[ \frac{1}{N} \tilde{X}^T g(S) \tilde{X} - \lambda({\sigma}^2_X) D\Big]^{-1}  \frac{1}{N}\tilde{X} g(S) Y.\]
It is straightforward that in a neighborhood of $\sigma_X^2$, we have $\lambda'(\sigma^2) = 1$. Therefore,
\[ \sqrt{N}[\lambda(\widehat{\sigma}_X^2  - \lambda(\sigma^2_X)] = O_p( \sqrt{N} [\widehat{\sigma}_X^2 - \sigma_X^2]),\]
\[ \sqrt{N}[\hat{\beta}(\widehat{\sigma}_X^2) - \hat{\beta}(\sigma_X^2)] = O_p\Big(\frac{1}{n_T}\Big). \]
It follows then
\[\sqrt{N} [\Xi(f)]^{-1/2}\Big(\hat{\beta}(\widehat{\sigma}_X^2)  - \beta \Big) \stackrel{D}{\longrightarrow}\calN(0, I_p).\]


\section{Proof of Theorem \ref{thm:adequacy_test}}\label{sec:proof_residual}
Consider the decomposition 
\[ \hat{\epsilon} =  X(\beta - \hat{\beta}) -  \sigma_X \Sigma^{1/2} W D^{1/2} \hat{\beta} + \Sigma^{1/2}\zeta,\]
where $\zeta$ is a vector of iid $\calN(0,1)$ entries and $W$ is a $N\times p$ matrix of iid $\calN(0,1)$ entries.  Let $G = [f(S)]^{-1}$, we have
\begin{align*}
\frac{1}{N}&\hat{\epsilon}^T G\hat{\epsilon} - \frac{1}{N}\tr[G\Sigma] - \sigma_X^2\hat{\beta}^TD\hat{\beta}\frac{1}{N}\tr[G\Sigma]  \\
= &\frac{\sigma_X^2}{N} \hat{\beta}^T D^{1/2}\Big[ W^T \Sigma^{1/2}G\Sigma^{1/2} W - \tr[G\Sigma] I_p\Big]D^{1/2}\hat{\beta}\\
&+ \frac{1}{N}\Big[\zeta^T \Sigma^{1/2}G\Sigma^{1/2}\zeta - \tr[G\Sigma] \Big]\\
&-\frac{2\sigma_X}{N} \zeta^T \Sigma^{1/2}G\Sigma^{1/2} W D^{1/2}\hat\beta\\
&+\frac{1}{N} (\hat{\beta} -\beta)^T X^TG X (\hat{\beta} -\beta) \\
&+ \frac{2}{N} \zeta^T \Sigma^{1/2}G X(\beta -\hat{\beta})\\
&+ \frac{2\sigma_X}{N} \hat{\beta}^T D^{1/2}W^T \Sigma^{1/2} G X (\hat\beta -\beta).
\end{align*}
Since $N^{-1} X^TG X = O_p(1)$ and from Theorem \ref{thm:consistency_a_b}, $(\hat{\beta} - \beta) = O_p(N^{-1/2})$, we have
\[ \frac{1}{N} (\hat\beta - \beta)^T X^T GX (\hat{\beta}- \beta) = O_p(N^{-1}).\]
Also, due to Lemma \ref{lemma:concentration_quadratic_forms},
\[ \frac{2}{N} \zeta^T \Sigma^{1/2}G X X^T G \Sigma^{1/2} \zeta = O_p(1).\]
It follows that 
\[ \frac{2}{N} \zeta^T \Sigma^{1/2}G X(\hat{\beta} - \beta) = O_p(N^{-1}).\]
Due to analogous arguments
\[ \frac{2\sigma_X}{N}\hat{\beta}^T D^{1/2} W^T \Sigma^{1/2} G X (\hat{\beta} - \beta) = O_p(N^{-1}).\]
Moreover, due to Lemma \ref{lemma:concentration_quadratic_forms},
\[\frac{1}{N} W^T \Sigma^{1/2}G\Sigma^{1/2}W - \frac{1}{N}\tr[G\Sigma] I_p = O_p(N^{-1/2}).\]
It follows then 
\[ \frac{\sigma_X^2}{N} \hat{\beta} ^T D^{1/2} \Big[ W^T \Sigma^{1/2} G \Sigma^{1/2} W - \tr[G\Sigma] I_p \Big] D^{1/2}\hat{\beta} - \frac{\sigma_X^2}{N} \beta ^T D^{1/2} \Big[ W^T \Sigma^{1/2} G \Sigma^{1/2} W - \tr[G\Sigma] I_p \Big] D^{1/2}\beta  = O_p(N^{-1}).\]
Also, 
\[ \frac{1}{N} \zeta^T \Sigma^{1/2}G \Sigma^{1/2} W D^{1/2}\hat{\beta} -  \frac{1}{N} \zeta^T \Sigma^{1/2}G \Sigma^{1/2} W D^{1/2}\beta = O_p(N^{-1}).\]
We conclude that 
\begin{align*}
\frac{1}{N}&\hat{\epsilon}^T G\hat{\epsilon} -  (1+  \sigma_X^2\hat{\beta}^TD\hat{\beta})\frac{1}{N}\tr[G\Sigma]  \\
= &\frac{1}{N}[\zeta - \sigma_X W D^{1/2}\beta]^T \Sigma^{1/2}G \Sigma^{1/2}  [ \zeta - \sigma_X WD^{1/2} \beta ] \\
& - \frac{1}{N}\mE [\zeta - \sigma_X W D^{1/2}\beta]^T \Sigma^{1/2}G \Sigma^{1/2}  [ \zeta - \sigma_X WD^{1/2} \beta ]  + o_p(N^{-1/2}).
\end{align*}
Clearly, $\zeta - \sigma_X WD^{1/2}\beta$ is a $N\times 1$ vector whose entries are iid $\calN(0,  1+ \sigma^2_X \beta^T D\beta)$. Using Lindeberg's central limit theorem,
\begin{align*}
 \sqrt{N} \frac{N^{-1} \hat{\epsilon}^T G \hat{\epsilon} - (1+  \sigma_X^2{\beta}^TD{\beta}) N^{-1}\tr[G\Sigma]}{\sqrt{2 (1+ \sigma_X^2 \beta^T D\beta)^2 N^{-1}\tr[G\Sigma G\Sigma]}} \stackrel{D}{\longrightarrow}\calN(0,1).
\end{align*}

On the other hand, we have 
\[ \frac{1}{N}\tr[G\Sigma] = \Pi(f) + o_p(N^{-1/2}).\] 
\[ \frac{1}{N}\tr[G\Sigma G\Sigma ]= K(f) + o_p(1).\]
\[ \widehat{\sigma}_X^2 - \sigma_X^2 = O_p(N^{-1/2}).\]
\[ \hat{\beta} = \beta + O_p(N^{-1/2}).\]
Therefore, using Slutsky's theorem, we conclude that 
 \[ \frac{Q(\hat{\epsilon})}{\sqrt{V_Q}} \stackrel{D}{\longrightarrow} \calN(0,1), \quad \mbox{where}\quad  V_Q  = \frac{2}{N} (1+ \widehat{\sigma}_X^2 \hat{\beta}^T D \hat{\beta}) K(f).\]

\end{document}